\documentclass[journal,twoside,web]{IEEEcolor}
\usepackage{generic}
\usepackage{cite}
\usepackage{amsmath,amssymb,amsfonts}
\usepackage{graphicx}
\usepackage{algorithm,algorithmic}
\usepackage{hyperref}
\hypersetup{hidelinks=true}
\usepackage{textcomp}
\usepackage{lipsum}
\usepackage{dsfont}
\usepackage{empheq,mathtools}
\usepackage{tikz}
\usepackage{circuitikz}
\usepackage{subcaption}

\def\BibTeX{{\rm B\kern-.05em{\sc i\kern-.025em b}\kern-.08em
    T\kern-.1667em\lower.7ex\hbox{E}\kern-.125emX}}

\newtheorem{secthm}{Theorem}[section]

\newtheorem{seclem}[secthm]{Lemma}

\newtheorem{secprop}[secthm]{Proposition}
\newtheorem{secprob}[secthm]{Problem}
\newtheorem{secdefn}[secthm]{Definition}
\newtheorem{secrem}[secthm]{Remark}
\newtheorem{secasm}[secthm]{Assumption}

\newcommand{\rank}{{\operatorname{rank}}}
\newcommand{\tr} { {\operatorname{tr}}}

\newcommand{\Cov} {\operatorname{Cov}}
\newcommand{\bE} { {\mathbb E}}

\newcommand{\bR} { {\mathbb R}}
\newcommand{\bS} { {\mathbb S}}
\newcommand{\bZ} { {\mathbb Z}}

\newcommand{\fF} { {\mathbf F}}

\newcommand{\cN} { {\mathcal N}}

\newcommand{\esssup}{\operatorname*{ess\,sup}}
\renewcommand{\theenumi}{\roman{enumi}}

\def\red{\hfill $\lhd$}

\begin{document}

\title{\LARGE \bf
 Design of Privacy-Preserving Estimators Under\\ Prior Information Exposure
% State Privacy in Dynamical Systems via Pointwise Maximal Leakage
% delete or comment-out the following line before submission
}

\author{%%%% author names
    Le Liu, Yu Kawano, \it{Senior Member, IEEE}, and Ming Cao, \it{Fellow, IEEE}% delete this line if not needed
    % duplicate the line above as many times as needed to list all authors
    \\%%%% author affiliations
    %\IEEEauthorblockA{\textit{ Engineering and Technology institute Groningen,University of Groningen, Nijenborgh 4, 9747 AG  Groningen, The Netherlands}}\\% first affiliation
    %\IEEEauthorblockA{\textit{Graduate School of Advanced Science and Engineering, Hiroshima University, Higashi-Hiroshima 739-8527, Japan}}\\% delete this line if not needed
    % duplicate the line above as many times as needed to list all affiliations
    %%%% corresponding author contact details
     \thanks{The work of Liu and Cao was supported in part by the Netherlands Organization for Scientific Research (NWO-Vici-19902). The work of Liu is partly supported by the China Scholarship Council. The work of Yu Kawano was supported in part by JSPS KAKENHI Grant Number JP25K22804.}
\thanks{Le Liu and Ming Cao are with the Faculty of Science and Engineering, University of Groningen, 9747 AG Groningen, The Netherlands {\tt \small \{le.liu, m.cao\}@rug.nl}}
    \thanks{Yu Kawano is with the Graduate School of Advance Science and Engineering, Hiroshima University, Higashi-Hiroshima 739-8527, Japan 
        {\tt\small ykawano@hiroshima-u.ac.jp}}
}
\maketitle

\begin{abstract}
State estimation is widely used in networked dynamical systems, but released estimates may inadvertently reveal sensitive information about private state variables. Existing privacy-preserving estimation frameworks often do not explicitly account for adversarial prior knowledge. Motivated by this gap, we study privacy-preserving state estimation for discrete-time linear stochastic systems under prior information exposure using pointwise maximal leakage (PML). We first reveal a privacy anti-optimality property of the Kalman filter: in the absence of additive release noise, the steady-state Kalman filter maximizes PML privacy leakage among the considered class of stable linear estimators. We further establish a nontrivial lower bound on the steady-state estimation error enforced by $(\varepsilon,\delta)$-PML privacy. We then derive sufficient release-noise conditions that guarantee $(\varepsilon,\delta)$-PML privacy, first for the Kalman filter and then for stable linear estimators. Finally, a convex semidefinite program is formulated for a joint design of the estimator gain and release noise covariance, providing a tractable privacy-preserving estimator with an explicit upper bound on the steady-state estimation error. Simulations on a source--load power system demonstrate the proposed privacy--utility tradeoff.
\end{abstract}

\begin{keywords}
Privacy, Estimator Design, Pointwise Maximal Leakage,  Kalman filter, Convex Optimization 
\end{keywords}

\section{Introduction}

State estimation plays a central role in modern cyber--physical systems. In applications such as smart grids, industrial automation, intelligent transportation, and networked sensing, measurements are continuously collected and processed to infer internal variables for monitoring, prediction, and control~\cite{yu2016smart,venancio2023cps,sisinni2018industrial}. However, released estimates may reveal information beyond the intended quantity of interest. In dynamical systems, correlations among state components can cause an estimate of one variable to disclose sensitive information about another. This risk is particularly relevant in networked infrastructures, where estimates are routinely transmitted, shared, and stored across multiple entities.

Privacy-preserving estimation has therefore attracted increasing attention. Unlike static data-release problems, privacy leakage in dynamical systems depends not only on the released signal but also on the estimator dynamics. Hence, estimator design itself becomes an important component of privacy protection. Differential privacy and its variants have been widely used in estimation and control problems~\cite{dwork2006calibrating,dwork2006our,dwork2016concentrated,mironov2017renyi,dong2022gaussian,le2013differentially,zhao2014achieving,liu2026fairness_robotics}. Differential privacy provides a distribution-independent guarantee and does not require specifying an adversarial prior in its definition. In contrast, PML explicitly quantifies privacy loss through the change from the prior distribution of the secret to its posterior distribution after a particular observation. Such prior information is often available in dynamical systems, since the system model, statistics, or steady-state behavior may be publicly known or learned from historical data.

This motivates the use of prior-aware privacy notions, such as mutual-information privacy~\cite{wang2016relation}, local information privacy~\cite{jiang2021context}, Fisher-information privacy~\cite{farokhi2019ensuring}, and pointwise maximal leakage (PML)~\cite{saeidian2023pointwise_tit}. Among them, PML is particularly attractive for linear stochastic systems because it admits explicit covariance-based characterizations in the Gaussian setting~\cite{liu2025privacy}. It therefore provides a tractable way to quantify how much additional information a released estimate reveals about a private state variable beyond the adversary's prior information.

\smallskip

\subsubsection*{Literature Review} 
Privacy-preserving estimation has been studied from both static and dynamical perspectives. In static settings, representative formulations have included Fisher-information-based privacy design \cite{ke2025privacy, ke2026privacy}, context-aware local information privacy \cite{jiang2021context}, and pointwise maximal leakage (PML) \cite{saeidian2023pointwise_tit,grosse2024extremal}. In particular, \cite{ke2025privacy} has developed additive-noise mechanisms by minimizing the Fisher information available to an estimator, while \cite{jiang2021context} has incorporated prior information into local privacy guarantees. However, the resulting privacy--utility design does not admit a simple closed-form characterization in general.

In parallel, privacy-preserving estimation and control for dynamical systems have been investigated largely through differential privacy and output perturbation, including differentially private filtering \cite{le2013differentially}, differential privacy in control and network systems \cite{Jorge2016,liu2025statistical, liu2025initial, liu2025design}, privacy-preserving average consensus and distributed algorithms \cite{mo2016privacy,nozari2017differentially,wang2019privacy, liu2025rendezvous}, privacy-preserving dynamic controller design \cite{Yu2020,Yu2021}, and privacy-aware optimization and online learning \cite{wang2023tailoring,wang2024robust,chen2024local}. Specifically, differentially private filtering has been investigated in \cite{le2013differentially}, where privacy guarantees have been incorporated into estimators for dynamical systems. Privacy-preserving dynamic controller design has been considered in \cite{Yu2020}, where controllers have been synthesized to protect sensitive signals while maintaining closed-loop performance. In estimation-oriented applications, cooperative privacy-preserving state estimation in smart grids has been studied in \cite{kim2011cooperative}, privacy-preserving distributed Kalman filtering has been developed in \cite{moradi2022privacy}, and optimal privacy-aware state estimation has recently been investigated in \cite{weng2025optimal}. In \cite{liu2025privacy}, the relationship between output privacy and Kalman filtering has been studied, and a privacy-dependent lower bound on the estimation error has been established. Broad overviews of privacy in control and dynamical systems can be found in \cite{han2018privacy,nekouei2019information,yan2024survey}.

Despite these developments, three questions remain open. First, it is unclear how the gain of a causal estimator affects prior-aware privacy leakage in steady state. Second, the privacy behavior of the Kalman filter under prior information exposure has not been characterized: although it is optimal for estimation, it may be unfavorable for privacy. Third, existing PML-based privacy conditions do not directly yield a tractable co-design method for the estimator gain and release noise covariance. These questions motivate the analysis and synthesis framework developed in this paper.

% Although these works provide valuable insights into privacy-preserving estimation and control, the role of adversarial prior knowledge in dynamical systems remains insufficiently understood. In many applications, an adversary may infer statistical information about the private variable from the system model and its steady-state behavior. Such prior information is not explicitly captured by most existing estimation frameworks, particularly differential privacy, whose guarantees are largely distribution-independent. In contrast, PML provides a prior-aware worst-case privacy notion and admits tractable characterizations in the Gaussian setting, making it well suited for linear stochastic systems. However, prior-aware privacy studies have mainly focused on static data-release problems. To the best of our knowledge, how estimator dynamics affect PML privacy leakage in stochastic dynamical systems has not been systematically studied. This paper addresses this gap by analyzing the impact of the estimator gain on steady-state PML privacy and developing a privacy-preserving estimator synthesis method. 
\smallskip

\subsubsection*{Contributions} 

In this paper, we study privacy-preserving estimation for discrete-time linear stochastic systems under PML. A linear function of the state is regarded as the secret variable, while another linear function is the quantity to be estimated and released through a linear time-invariant estimator with additive Gaussian perturbation. We first analyze the steady-state Kalman filter from a PML perspective. Under a structural condition in which the secret variable lies in the linear span of the estimate and no additional release noise is introduced, we show that it is the least private estimator within the considered class of stable linear estimators. Furthermore, under the same structural condition, we establish a nontrivial lower bound on the steady-state estimation error imposed by the $(\varepsilon,\delta)$-PML privacy constraint. 

Building on this analysis, we develop privacy-preserving estimator design methods. We derive sufficient conditions on the release noise covariance that guarantee $(\varepsilon,\delta)$-PML privacy for both the Kalman filter and general stable linear estimators of the form~\eqref{eq:estimator}. We then formulate a convex semidefinite program to jointly design the estimator gain and the release noise covariance, yielding a tractable privacy-aware estimator with an explicit upper bound on the steady-state estimation error. The main contributions of this paper are summarized as follows:
\begin{enumerate}
    \item We formulate a steady-state PML privacy problem for stable linear time-invariant estimators, where a private linear state variable must be protected while a possibly different linear state variable is estimated and released. The formulation is motivated by a source--load power system in which released current estimates may reveal private user consumption information.

    \item We reveal a privacy anti-optimality property of the Kalman filter. Specifically, under the structural condition that the secret variable is a linear function of the released variable and in the absence of release noise, the steady-state Kalman filter is the least private estimator among stable linear estimators of the form~\eqref{eq:estimator}, in the sense that it maximizes the required PML privacy budget for every fixed failure probability.

    \item Under the same structural condition, we show that $(\varepsilon,\delta)$-PML privacy imposes a nontrivial lower bound on the steady-state estimation error, independently of how the estimator gain and release noise covariance are chosen.

    \item We derive sufficient release-noise covariance conditions that guarantee $(\varepsilon,\delta)$-PML privacy, first for the Kalman filter and then for general stable linear time-invariant estimators.

    \item We propose a convex semidefinite program for jointly designing the estimator gain and release noise covariance. Any feasible solution guarantees $(\varepsilon,\delta)$-PML privacy and provides an explicit upper bound on the steady-state estimation error.
\end{enumerate}

\subsubsection*{Organization}

The remainder of this paper is organized as follows. Section~\ref{sec:prob} presents a motivating example, introduces the problem formulation, and defines the adopted PML privacy notion. Section~\ref{sec:Kalman} analyzes the privacy limitations of the Kalman filter. Section~\ref{sec:filter} develops privacy-preserving estimator design methods and presents a semi-definite-programming-based synthesis framework. Section~\ref{sec:sim} illustrates the proposed approach through numerical simulations. Section~\ref{sec:con} concludes this paper.

%%%%%%%%%%%%%%%%%%%%%%%%%%%%%%%%%%%%%%%%%%%%%%%%%%%%%%%%%%%%%%%%%%%%%%%
%%%%%%%%%%%%%%%%%%%%%%%%%%%%%%%%%%%%%%%%%%%%%%%%%%%%%%%%%%%%%%%%%%%%%%%
%%%%%%%%%%%%%%%%%%%%%%%%%%%%%%%%%%%%%%%%%%%%%%%%%%%%%%%%%%%%%%%%%%%%%%%

\smallskip

\subsubsection*{Notation}
The sets of real numbers, integers and nonnegative integers are denoted by $\bR$, $\bZ$ and $\bZ_{+}$, respectively. 
For $n\in\bZ_{+}$, let $\bS^{n}$ be the set of $n\times n$ real symmetric matrices, and let $\bS^{n}_{+}$ and $\bS^{n}_{++}$ denote the positive semidefinite and positive definite cones, respectively. For integers $t \le T$, we write $ x_{t:T} := \begin{bmatrix} x_t^\top & x_{t+1}^\top & \cdots & x_T^\top \end{bmatrix}^\top$
for the stacked release vector over the time interval from $t$ to $T$.
For $P,Q\in\bS^{n}$, $P \succeq Q$ (resp. $P \succ Q$) means $P-Q\in\bS^{n}_{+}$ (resp. $P-Q\in\bS^{n}_{++}$). For $P \in \bS^{n}$,
$\lambda_{\min}(P) := \lambda_1(P) \le \lambda_2 (P) \le \cdots \le \lambda_n(P) =:\lambda_{\max}(P)$ denote its eigenvalues. For $A \in \bR^{n \times n}$, $\rank(A)$ and $\operatorname{tr}(A)$ denote its rank and trace, respectively; $\rho(A)$ denotes its spectral radius; $A^\dagger$ denotes its Moore--Penrose pseudoinverse; $\operatorname{Im}(A)$ and $\operatorname{Ker}(A)$ denote its image and kernel, respectively. For a subspace $\mathcal V$, $\mathcal V^\perp$ denotes its orthogonal complement.
The identity matrix of dimension $n$ is denoted by $I_n$. For vectors and matrices, the $2$-norm and the Frobenius norm are denoted by $\|\cdot \|_2 $ and $\|\cdot\|_{F}$, respectively. 

The expectation of a random variable is denoted by $\bE[\cdot]$. For random vectors $x$ and $y$ with finite second moments, $\Cov(x):=\bE[(x-\bE[x])(x-\bE[x])^\top]$ denotes the covariance matrix of $x$, and $\Cov(x,y):=\bE[(x-\bE[x])(y-\bE[y])^\top]$ denotes the cross-covariance matrix of $x$ and $y$. A Gaussian distribution with mean $\mu$ and covariance $\Sigma\succeq0$ is written as $\cN(\mu,\Sigma)$. The chi–squared distribution with $\ell$ degrees of freedom is denoted by $\chi^{2}_{\ell}$; its cumulative distribution function (CDF) is denoted by $\fF_{\chi_{\ell}^2}(\cdot)$, and its quantile function is $\fF^{-1}_{\chi_{\ell}^2}(\cdot)$. When $\ell = 0$, we use the convention 
$\fF^{-1}_{\chi^2_0}(1-\delta)=0$.

%%%%%%%%%%%%%%%%%%%%%%%%%%%%%%%%%%%%%%%%%%%%%%%%%%%%%%%%%%%%%%%%%%%%%%%
%%%%%%%%%%%%%%%%%%%%%%%%%%%%%%%%%%%%%%%%%%%%%%%%%%%%%%%%%%%%%%%%%%%%%%%
%%%%%%%%%%%%%%%%%%%%%%%%%%%%%%%%%%%%%%%%%%%%%%%%%%%%%%%%%%%%%%%%%%%%%%%

%%%%%%%%%%%%%%%%%%%%%%%%%%%%%%%%%%%%%%%%%%%%%%%%%%%%%%%%%%%%%%%%%%%%%%%
%%%%%%%%%%%%%%%%%%%%%%%%%%%%%%%%%%%%%%%%%%%%%%%%%%%%%%%%%%%%%%%%%%%%%%%
%%%%%%%%%%%%%%%%%%%%%%%%%%%%%%%%%%%%%%%%%%%%%%%%%%%%%%%%%%%%%%%%%%%%%%%

\section{Motivation and  Problem Formulation}
\label{sec:prob}
% Privacy preservation is an important issue in dynamical systems when releasing state estimates may inadvertently disclose sensitive user information. Differential privacy (DP)~\cite{dwork2006differential, le2013differentially} provides a systematic method for designing such noise and quantifying privacy. However, DP does not explicitly incorporate adversarial prior knowledge, even though such information is often publicly available in dynamical systems. In this paper, we adopt \emph{pointwise maximal leakage} (PML) and develop a systematic framework for designing privacy-preserving estimators for state estimates in the presence of public prior information.

% In this section, we first present a motivating example, and then introduce the definition of PML.

In this section, we first present a motivating source--load example and
then formulate the steady-state PML privacy problem for stable linear time-invariant estimators.

%%%%%%%%%%%%%%%%%%%%%%%%%%%%%%%%%%%%%%%%%%%%%%%%%%%%%%%%%%%%%%%%%%%%%%%
%%%%%%%%%%%%%%%%%%%%%%%%%%%%%%%%%%%%%%%%%%%%%%%%%%%%%%%%%%%%%%%%%%%%%%%
%%%%%%%%%%%%%%%%%%%%%%%%%%%%%%%%%%%%%%%%%%%%%%%%%%%%%%%%%%%%%%%%%%%%%%%

\subsection{Motivating Example}
\label{sec:example}

\begin{figure*}[t]
\centering
% ---------------- (a) abstract graph ----------------
\begin{subfigure}{0.26\textwidth}
\centering
\begin{circuitikz}[european resistors, american currents, scale=1]
    \draw
    (0,0) node[circ, scale=0.6, fill=black]{} coordinate(S)
    to[R, *-* , l=$R$] (2.6,0) node[circ, scale=0.6, fill=black]{} coordinate(L);

    \node[above=2mm] at (S) {source};
    \node[above=2mm] at (L) {load};
\end{circuitikz}
\caption{Abstract source--load interconnection.}
\end{subfigure}
\hfill
% ---------------- (b) source expansion ----------------
\begin{subfigure}{0.30\textwidth}
\centering
\begin{circuitikz}[european resistors, american currents, scale=1]
    \draw
    (0,2) node[circ]{} coordinate(S)
          to[short] (2.6,2);

    \node[left] at (0,2) {source};

    % top signal
    \draw[->] (0,3.2) -- (0,2.1);
    \node[left] at (0,2.8) {$u(t)$};

    % source-side dynamics
    \draw
    (S) to[C, l_=$C_1$] (0,0)
        node[ground]{};

    % current label
    \draw[->] (1.0,2.35) -- (2.1,2.35);
    \node at (1.55,2.65) {$i_s(t)$};
\end{circuitikz}
\caption{Expanded source-side model.}
\end{subfigure}
\hfill
% ---------------- (c) load expansion ----------------
\begin{subfigure}{0.34\textwidth}
\centering
\begin{circuitikz}[european resistors, american currents, scale=1]
    % top node
    \draw
    (0,2) node[circ]{} coordinate(L) to[short] (4.2,2);

    \node[above] at (0.2,2) {load};

    % load-side capacitor
    \draw
    (L) to[C, l_=$C_2$] (0,0)
    node[ground]{};

    % baseline load
    \draw
    (2.4,2) to[R, l_=$R_L$] (2.4,0)
    node[ground]{};

    % user consumption / demand branch
    \node[draw, minimum width=1.35cm, minimum height=0.75cm, align=center] (D) at (4.2,1.0) {user load\\$d(t)$};
    \draw (4.2,2) -- (D.north);
    \draw (D.south) -- (4.2,0);
    \node[ground] at (4.2,0) {};
\end{circuitikz}
\caption{Expanded load-side model.}
\end{subfigure}

\caption{Source--load representation and the corresponding source-side and load-side internal dynamics. The branch labeled $d(t)$ denotes the user-driven consumption component.}
\label{fig:circuit}
\end{figure*}
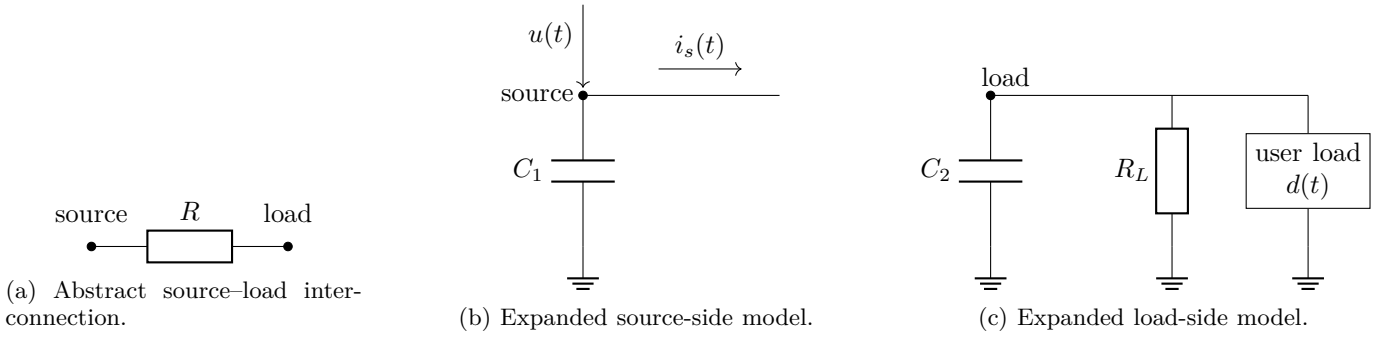

\begin{figure*}[!t]
    \centering
    \includegraphics[width= 0.8\textwidth]{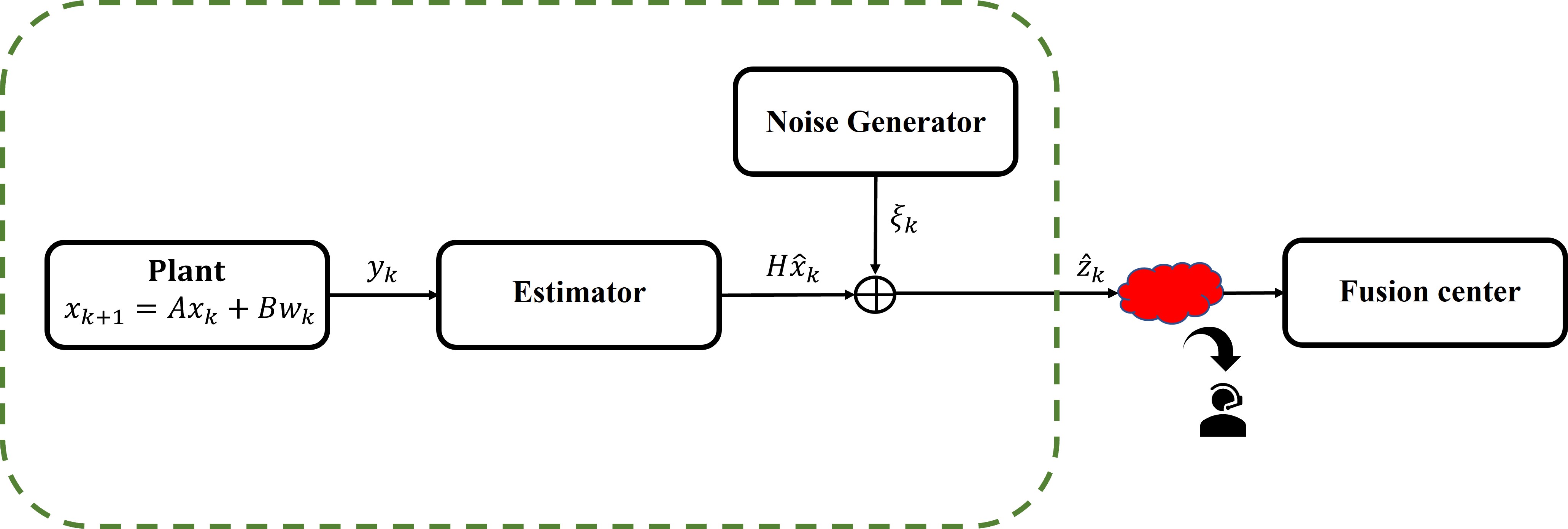}
    \caption{Mechanism diagram: The plant output $y_k$ is processed by the estimator to produce $H\hat{x}_k$, and the noise generator then adds privacy noise $\xi_k$ to form the released signal $\hat{z}_k$. The fusion center receives $\hat{z}_k$, while an eavesdropper may intercept the transmitted signal over the communication channel. The local components inside the green dashed box are assumed secure and not subject to eavesdropping.}
    \label{fig:arch}
\end{figure*}

As a motivating example, consider the source--load power system shown in Fig.~\ref{fig:circuit} \cite{ZHAO201518}. A source node and a load node are connected through a line of resistance $R$. The source is equipped with a local voltage controller that regulates the source-side voltage $v_s(t)$ around a prescribed reference $v_{\mathrm{ref}}$. The load side consists of a baseline resistive demand $R_L$, a capacitor $C_2$, and a user-driven consumption component $d(t)$. Because $d(t)$ is directly related to the user's energy usage, it is regarded as private information.

Let $v_s(t)$ and $v_\ell(t)$ denote the source-side and load-side voltages, respectively. The system is modeled by the It\^o stochastic differential equations:
\begin{align*}
C_1\, \mathrm{d}v_s(t)
&= \left(-\frac{v_s(t)-v_\ell(t)}{R}+u(t)\right)\mathrm{d}t + b_1\, \mathrm{d}W_1(t),\\
C_2\, \mathrm{d}v_\ell(t)
&= \left(\frac{v_s(t)-v_\ell(t)}{R}-\frac{v_\ell(t)}{R_L}-d(t)\right)\mathrm{d}t + b_2\, \mathrm{d}W_2(t),\\
\mathrm{d}d(t)
&= -\lambda d(t)\,dt + b_3\, \mathrm{d}W_3(t),
\end{align*}
where $\lambda>0$, the constants $b_1,b_2,b_3$ are diffusion intensities, and $W_1,W_2,W_3$ are mutually independent standard Wiener processes. In contrast to the constant-load model $\mathrm{d}d(t)=0$ in~\cite{Yu2020,ZHAO201518}, the consumption state $d(t)$ is modeled here as a stochastic process~\cite{Verdejo2016AppliedEnergy,Verdejo2019SDE}, allowing random but persistent demand variations. Nevertheless, the constant-load model can be recovered approximately by taking $\lambda$ and $b_3$ sufficiently small.

To represent the internal source-side voltage regulation, we introduce a controller state $\eta(t)$ and define control input $u(t)$ by
\begin{align*}
    \left\{
\begin{aligned}
\dot{\eta}(t) &= v_s(t)-v_{\mathrm{ref}},\\
u(t) &= -k_p\bigl(v_s(t)-v_{\mathrm{ref}}\bigr)-k_i\eta(t),
\end{aligned}
\right.
\end{align*}
where $k_p>0$ and $k_i>0$ are the proportional and integral voltage-control gains, respectively; $v_{\mathrm{ref}}$ is the reference voltage value. This controller regulates the source-side voltage toward the prescribed reference value.

{For privacy analysis, it is standard to deal with discrete-time systems \cite{Yu2020, liu2025design}. A temporal discretization of the closed-loop system yields} 
%The closed-loop continuous-time system can be discretized and written in a compact form as:
\begin{align*}
x_{k+1}=Ax_k+Bw_k,
\end{align*}
where $x_k :=
\begin{bmatrix}
v_{s,k}-v_s^\ast &
v_{\ell,k}-v_\ell^\ast &
d_k-d^\ast &
\eta_k-\eta^\ast
\end{bmatrix}^{\top}$
denotes the augmented deviation state, with $v_s^\ast = v_{\mathrm{ref}},
v_\ell^\ast = \frac{R_L}{R+R_L}v_{\mathrm{ref}},
d^\ast = 0,
\eta^\ast = -\frac{v_{\mathrm{ref}}}{k_i(R+R_L)}$. The sequence $\{w_k\}$ is an independent standard Gaussian process, and the matrices $A$ and $B$ are obtained from the closed-loop continuous-time dynamics.

As the energy provider relies on source-current information to monitor power flow, a sensor measures the line current
$$
y_k = i_{s,k}+\vartheta_k,
$$
where $i_{s,k} = (v_{s,k} - v_{\ell,k}) / R$ is the line current and  the sequence $\{\vartheta_k\}$ is an independent standard Gaussian process. The measurement is processed locally by an estimator to reduce the computational burden at the fusion center, producing the current estimate $\hat{i}_{s,k}$. The estimate $\hat{i}_{s,k}$ is then transmitted to the fusion center for monitoring. Although $\hat{i}_{s,k}$ is intended to provide useful operational information, it may also carry information about the user's private energy consumption $d_k$.

A key feature of this setting is that the user consumption state $d_k$ may already be partially predictable before any estimate is released. In practice, the system model, operating conditions, and historical data may be available to an eavesdropper. Consequently, the steady-state distribution of $d_k$ can serve as prior information. The privacy risk of releasing $\hat i_{s,k}$ should therefore be measured not only by the information contained in the released signal itself, but also by the additional information it reveals beyond this prior. This motivates the use of PML as the privacy metric in this paper.

%

%%%%%%%%%%%%%%%%%%%%%%%%%%%%%%%%%%%%%%%%%%%%%%%%%%%%%%%%%%%%%%%%%%%%%%%
%%%%%%%%%%%%%%%%%%%%%%%%%%%%%%%%%%%%%%%%%%%%%%%%%%%%%%%%%%%%%%%%%%%%%%%
%%%%%%%%%%%%%%%%%%%%%%%%%%%%%%%%%%%%%%%%%%%%%%%%%%%%%%%%%%%%%%%%%%%%%%%

\subsection{Problem Formulation}

For generality, we consider the following linear time-invariant discrete-time system:
\begin{align}\label{eq:LTI}
    \left\{
    \begin{aligned}
        x_{k+1} &= A x_k + B w_k,\\
        y_k &= C x_k + D v_k,
    \end{aligned}
    \right.
\end{align}
where $x_k \in \mathbb{R}^n$ is the system state and $y_k \in \mathbb{R}^m$ is the measurement output. The noise processes $\{w_k\}$ and $\{v_k\}$ are mutually independent and identically distributed (i.i.d.) zero-mean Gaussian sequences satisfying
\begin{align*}
    \bE[w_k w_k^\top] = I_p, \qquad \bE[v_k v_k^\top] = I_q,
\end{align*}
and are also independent of the initial state $x_0$, where $x_0$ is a random Gaussian variable independent of $\{w_k\}$ and $\{v_k\}$. The matrices $A$, $B$, $C$, and $D$ have compatible dimensions. 

Within this system model, we distinguish between the variable used for estimation and the variable that is considered private. Specifically, the following linear function of the state is considered as the privacy-sensitive variable:
\begin{align}
\label{eq:secret}
\zeta_k = F x_k,
\end{align}
where $F \in \bR^{r \times n}$. 
%We next introduce the estimation and release mechanism. In this setting, $\zeta_k$ is the private quantity to be protected.
%while the fusion center aims to estimate a different linear function of the state, given by
%Since linearly dependent rows of $F$ do not provide additional information about the secret variable $\zeta_k$, it is natural to remove this redundancy by imposing the following assumption.

{For the operation purpose, the fusion center requires to receive}
\begin{align*}
z_k = H x_k,
\end{align*}
where $H \in \bR^{s \times n}$. {To send this information, the plant needs to estimate it by the measured output~$y_k$. We suppose that the estimation of~$z_k$ is constructed from the state estimation obtained by the standard observer:}
%To this end, the plant is equipped with a local estimator of the form
\begin{align}
\label{eq:estimator}
\hat{x}_{k+1}
= A \hat{x}_k + K\bigl(y_{k+1}-CA\hat{x}_k\bigr),
\end{align}
where $K \in \bR^{n \times m}$ is the estimator gain. Based on the state estimate $\hat{x}_k$, the quantity released to the fusion center is an estimate of $z_k$ corrupted by additive noise, i.e.,
\begin{align}
\label{eq:noise}
    \hat{z}_k = H\hat{x}_k + \xi_k,
\end{align}
 where the sequence $\{\xi_k\}$, merely added for the purpose of protecting from~$\zeta_k$ begin identified, is an i.i.d. zero-mean Gaussian sequence with covariance matrix $\Xi$, independent of $\{w_k\}$, $\{v_k\}$, and the initial state $x_0$. The resulting privacy-preserving release architecture is illustrated in Fig.~\ref{fig:arch}.  %{\color{red} The estimate of $z_k$ is computed by a trusted local estimator located at the plant side. This architecture reduces the computational burden at the remote fusion center by performing part of the filtering locally before transmission, which is standard in remote estimation \cite{ding2020remote,zhang2025remote}.}{\color{blue}YK: I'm confused. Why does the plant needs to estimate~$z_k$? It can directly send its output~$y_k$ (after adding noise) to the fusion center.}

The released signal $\hat z_k$ is required to serve two competing objectives. It should provide useful information for estimating the utility variable $z_k$, while limiting the disclosure of the privacy-sensitive variable $\zeta_k$. Thus, the central design problem is to choose the estimator gain $K$ and the release-noise covariance $\Xi$ so as to balance estimation performance and privacy.

\begin{secprob}
\label{prob:1}
For the estimator in the form of~\eqref{eq:estimator} and~\eqref{eq:noise}, how can the estimator gain $K$ and the noise covariance $\Xi$ be designed so that $\hat z_k$ remains useful {as an estimation of} $z_k$, while limiting the disclosure of the privacy-sensitive variable $\zeta_k$ to an adversary with prior knowledge of the system model and steady-state distribution?
\red
\end{secprob}

To answer Problem~\ref{prob:1}, we first introduce assumptions that make the steady-state privacy analysis well defined and specify the adversary's prior information rigorously. We then present the PML privacy criterion used to quantify disclosure of the privacy-sensitive variable.

%We impose the following assumptions.

\subsection{Standing Assumptions}
We study Problem~\ref{prob:1} under the following assumptions.

\begin{secasm}
\label{asm:rank_F}
The matrix $F$ in \eqref{eq:secret} has full row rank.
\red
\end{secasm}

Assumption~\ref{asm:rank_F} removes redundant components of the privacy-sensitive. %variable, while Assumption~\ref{asm:2} ensures that the state process admits a unique nondegenerate steady-state Gaussian distribution. Specifically,

\begin{secasm}
\label{asm:2}
The pair $(A,B)$ is controllable and the system matrix $A$ is Schur stable, i.e.,
\begin{align*}
\rho(A)<1. \tag*{\red}
\end{align*}
\end{secasm}

Assumption~\ref{asm:2} guarantees the existence of a unique steady-state Gaussian distribution for system~\eqref{eq:LTI}. Indeed, under Assumption~\ref{asm:2}, the asymptotic distribution of $x_k$ is uniquely determined by the system matrices $A$ and $B$, and its covariance matrix $\Sigma_\infty$ is the unique positive definite solution to the following discrete-time Lyapunov equation,
\begin{align}
\label{eq:steady}
    \Sigma_\infty = A\Sigma_\infty A^\top + BB^\top. 
\end{align}
If the estimator~\eqref{eq:estimator} is additionally stable, then there exists a random vector $(\zeta_\infty,\hat z_\infty)$ such that
\begin{align*}
    \mathcal{L}(\zeta_k,\hat z_k)\xRightarrow[k\to\infty]{}\mathcal{L}(\zeta_\infty,\hat z_\infty),
\end{align*}
where $\mathcal{L}(\zeta_\infty,\hat z_\infty)$ is the limiting joint distribution of $(\zeta_k,\hat z_k)$. We therefore evaluate privacy in steady state, where the long-run joint distribution of the secret and released variables is well defined. In addition, $\mathcal{L}(\zeta_\infty,\hat z_\infty)$ may be regarded as prior knowledge available to an adversary, because the system dynamics are often publicly known or can be inferred from historical data.

In some analysis, we also impose the following assumption.
\begin{secasm}
\label{asm:3}
The measurement-noise covariance satisfies~$DD^\top \succ 0$. \red
\end{secasm}

\subsection{Pointwise Maximal Leakage Privacy}

We now introduce the PML privacy criterion \cite{liu2025privacy, saeidian2023pointwise_tit} used to quantify the information leaked from the privacy-sensitive variable $\zeta_\infty$ through the released signal $\hat z_\infty$.

\begin{secdefn}
The pointwise maximal leakage from $\zeta_\infty$ to a realization $\hat{z}_\infty=\hat{z}$ is defined as
\begin{align}
\label{eq:PML}
    \ell(\zeta_\infty \to \hat{z})
    := \log \esssup_{\zeta \in \mathbb{R}^{r}}
    \frac{f_{\zeta_\infty \mid \hat{z}_\infty}(\zeta \mid \hat{z}_{\infty} = \hat{z})}{f_{\zeta_\infty}(\zeta)},
\end{align}
where $f_{\zeta_\infty \mid \hat{z}_\infty}$ denotes the conditional density function of $\zeta_\infty$ given the observation $\hat{z}_\infty$, and $f_{\zeta_\infty}$ denotes the marginal density function of $\zeta_\infty$.
\red
\end{secdefn}

The quantity $\ell(\zeta_\infty \to \hat z)$ in~\eqref{eq:PML} measures the worst-case increase from the prior density of $\zeta_\infty$ to the posterior density after observing a particular released value $\hat z_\infty=\hat z$. Thus, a larger value of $\ell(\zeta_\infty \to \hat z)$ corresponds to a larger privacy loss for that realization. Since the released signal $\hat z_\infty$ is random, the induced privacy loss $\ell(\zeta_\infty\to \hat z_\infty)$ is also a random variable. We therefore require this privacy loss to be below a prescribed level with high probability, leading to the following steady-state privacy notion.

\begin{secdefn}
\label{def:pml_privacy_sys}
The estimator characterized by~\eqref{eq:estimator} and~\eqref{eq:noise}, with gain $K$ and noise covariance $\Xi$, is said to preserve \emph{$(\varepsilon,\delta)$-PML privacy} for the secret variable $\zeta_\infty$ if
\begin{align}
\label{eq:PML_privacy_sys}   \mathbb{P}\bigl(\ell(\zeta_\infty \to \hat{z}_\infty) \le \varepsilon\bigr) \ge 1-\delta,
\end{align}
where $\varepsilon \ge 0$ and $\delta \in (0,1)$.
\red
\end{secdefn}

Here, $\varepsilon$ in~\eqref{eq:PML_privacy_sys} quantifies the allowable privacy leakage level, while $\delta$ in~\eqref{eq:PML_privacy_sys} specifies the probability with which this leakage level may be exceeded. Therefore, smaller values of $\varepsilon$ and $\delta$ correspond to a stronger privacy requirement.
\begin{secrem}
Definition~\ref{def:pml_privacy_sys} is formulated for real-time privacy, since the estimator releases one message at each time instant. Finite-horizon privacy can be treated similarly by stacking variables over time. For instance, if one wishes to protect $\zeta_k$ when the disclosed information is the release sequence $\hat z_{k:k+T-1}$, then one can regard $\hat z_{k:k+T-1}$ as a single stacked release vector and evaluate the leakage from $\zeta_k$ to that vector. Therefore, privacy of $\zeta_k$ against multiple releases can still be analyzed within the same framework.
\red
\end{secrem}
%%%%%%%%%%%%%%%%%%%%%%%%%%%%%%%%%%%%%%%%%%%%%%%%%%%%%%%%%%%%%%%%%%%%%%%
%%%%%%%%%%%%%%%%%%%%%%%%%%%%%%%%%%%%%%%%%%%%%%%%%%%%%%%%%%%%%%%%%%%%%%%
%%%%%%%%%%%%%%%%%%%%%%%%%%%%%%%%%%%%%%%%%%%%%%%%%%%%%%%%%%%%%%%%%%%%%%%

%%%%%%%%%%%%%%%%%%%%%%%%%%%%%%%%%%%%%%%%%%%%%%%%%%%%%%%%%%%%%%%%%%%%%%%
%%%%%%%%%%%%%%%%%%%%%%%%%%%%%%%%%%%%%%%%%%%%%%%%%%%%%%%%%%%%%%%%%%%%%%%
%%%%%%%%%%%%%%%%%%%%%%%%%%%%%%%%%%%%%%%%%%%%%%%%%%%%%%%%%%%%%%%%%%%%%%%

\section{Privacy Limitations of the Kalman Filter}
\label{sec:Kalman}
This section studies the PML privacy of the Kalman filter. We first recall the steady-state Kalman filter and formulate the corresponding released signal. We then derive a Gaussian characterization of the induced PML privacy loss and specialize it to the Kalman-filter-based release. Finally, we compare the Kalman filter with other stable linear estimators under the structural relation $F=LH$ and in the absence of release noise.

\subsection{Review of Kalman Filter}

It is well known that the Kalman filter provides the minimum mean-square error estimate for system~\eqref{eq:LTI}. In particular, the Kalman filter converges to the steady state under Assumption~\ref{asm:2}. Furthermore, the steady-state Kalman filter can be viewed as a special case of the estimator~\eqref{eq:estimator}, corresponding to the choice $K=K^{\mathrm{KF}}$, where $K^{\mathrm{KF}}$ is given by
\begin{align*}
K^{\mathrm{KF}}
=
P^- C^\top
\bigl(CP^-C^\top + DD^\top\bigr)^{-1},
\end{align*}
and $P^-$ denotes the steady-state \emph{prior} error covariance satisfying
\begin{align*}
P^-
&=
A P^- A^\top + BB^\top \\
&- AP^- C^\top
\bigl(CP^-C^\top + DD^\top\bigr)^{-1}
CP^-A^\top.
\end{align*}
The corresponding steady-state \emph{posterior} error covariance is denoted by $P^{\mathrm{KF}}$ and is obtained from the prior covariance $P^-$ through the steady-state Kalman update
\begin{align}
\label{eq:PKF}
P^{\mathrm{KF}}
&=
P^- - P^- C^\top
\bigl(CP^-C^\top + DD^\top\bigr)^{-1}
CP^-.
\end{align}

In the later analysis, we require positive definiteness of~$P^{\mathrm{KF}}$, which is guaranteed under mild assumptions.

\begin{seclem}
\label{lem:prior_pd}
Under Assumptions~\ref{asm:2} and~\ref{asm:3}, it follows that~$P^{\mathrm{KF}} \succ 0$.
\end{seclem}
\begin{proof}
Since $A$ is Schur stable under Assumption~\ref{asm:2}, the pair $(A,C)$ is detectable. Together with the controllability of $(A,B)$ and the condition $DD^{\top}\succ 0$, Theorem~13.1.3 in~\cite{lancaster1995algebraic} implies that the steady-state prior error covariance satisfies $P^{-}\succ 0$. Moreover, the steady-state posterior error covariance can be written as
\begin{align*}
P^{\mathrm{KF}} = 
\left((P^-)^{-1}+C^\top(DD^\top)^{-1}C\right)^{-1}.
\end{align*}
Since $(P^-)^{-1}\succ0$ and $C^\top(DD^\top)^{-1}C\succeq0$, it follows that $P^{\mathrm{KF}}\succ0$.
\end{proof}

% The positive definiteness of $P^{\mathrm{KF}}$ is important for the privacy analysis below, particularly in the absence of release noise.

%  If $P^{\mathrm{KF}}\succ0$, then every nonzero linear component of the state retains some uncertainty after observing the Kalman estimate, which is sufficient for privacy of the Kalman filter. We therefore first provide a condition that guarantees $P^{\mathrm{KF}}\succ0$.

% \begin{seclem}
% \label{lem:prior_pd}
% Suppose Assumption~\ref{asm:2} holds and $DD^\top \succ 0$. Then, it holds that
% \begin{align*}
% P^{\mathrm{KF}} \succ 0.
% \end{align*}
% \end{seclem}
% \begin{proof}
% See Appendix~\ref{app:lem}.
% \end{proof}

Let $\hat{x}_k^{\mathrm{KF}}$ denote the steady-state Kalman estimate. The corresponding released signal is
\begin{align}
\label{eq:kf_release}
    \hat{z}_k^{\mathrm{KF}} = H\hat{x}_k^{\mathrm{KF}} + \xi_k .
\end{align}
Since the Kalman estimation error has steady-state covariance $P^{\mathrm{KF}}$ and $\xi_k$ is independent of the estimator, the steady-state release error covariance is
\begin{align*}
\Cov(z_\infty-\hat z_\infty^{\mathrm{KF}})
=
HP^{\mathrm{KF}}H^\top+\Xi .
\end{align*}

% With this privacy criterion in place, the central objective is to understand whether estimation accuracy and PML privacy are compatible, and if not, how the estimator should be redesigned. This leads to the following problem. 

 % The Kalman filter release $\hat z_k^{\mathrm{KF}}$ remains optimal in the minimum mean-square error sense among releases of the form $H\hat x_k+\xi_k$ with the same release noise covariance $\Xi$. However, a more accurate released estimate may also carry more information about the secret variable. This motivates the following question: although the Kalman filter release is optimal for estimation, how much information does it reveal about $\zeta_k$?

% This section addresses the first part of Problem~\ref{prob:main}. We
% analyze the PML privacy induced by the steady-state Kalman filter
% and show that, under a structural condition{\color{blue}YK: what do you mean by a structural condition?}, it is shown to be the least private
% estimator among the considered stable linear estimators when no release
% noise is added. 

\subsection{Privacy Limitations}

Although the Kalman filter estimate is optimal for estimation, its privacy performance must be evaluated separately. We therefore next characterize the PML privacy loss induced by Gaussian steady-state releases.

To simplify the presentation of the forthcoming result, we introduce the following notation:
\begin{align}
    \label{eq:steady_zeta}
    &\Sigma_{\zeta\zeta} := \Cov(\zeta_\infty) = F\Sigma_{\infty}F^{\top}, \\
    \nonumber
&\Gamma_{\zeta\zeta} := \Cov(\zeta_\infty \mid \hat z_\infty).
\end{align}
The following lemma characterizes $(\varepsilon,\delta)$-PML privacy for jointly Gaussian steady-state variables.

\begin{seclem}
\label{lem:gaussian_pml_characterization}
Suppose that Assumptions~\ref{asm:rank_F}--\ref{asm:2} hold. For any finite $\varepsilon \ge 0$ and $\delta \in (0,1)$, the estimator characterized by~\eqref{eq:estimator} and~\eqref{eq:noise} with gain $K$ and noise covariance $\Xi$ preserves $(\varepsilon,\delta)$-PML privacy for the secret variable $\zeta_{\infty}$ if and only if $\Gamma_{\zeta\zeta} \succ 0$ and
\begin{align}
\label{eq:lem_privacy}
\mathbf{F}^{-1}_{\chi^2_\ell}(1-\delta)
\le
2\varepsilon-\log\det(\Sigma_{\zeta\zeta})+\log\det(\Gamma_{\zeta\zeta}),
\end{align}
where $\ell := \operatorname{rank}(\Sigma_{\zeta\hat z})$ and $
\Sigma_{\zeta\hat z} := \Cov(\zeta_\infty, \hat z_\infty)$.
\end{seclem}
\begin{proof}
See Appendix~\ref{app:lem1}.
\end{proof}

\begin{secrem}
    Compared with Theorem~3.3 in~\cite{liu2025privacy}, Lemma~\ref{lem:gaussian_pml_characterization} is tailored to the present estimator-release setting, where nondegeneracy of the released signal cannot be assumed a priori. Indeed, the covariance of $\hat z_\infty$ depends on the design variables $K$ and $\Xi$, as well as on $H$, and may be singular for some admissible designs. The lemma accounts for this possibility through $\Gamma_{\zeta\zeta}\succ0$ as the condition for a finite PML guarantee. \red    
\end{secrem}

Since both $\ell$ and $\Gamma_{\zeta\zeta}$ depend on the estimator gain $K$, Lemma~\ref{lem:gaussian_pml_characterization} shows that the choice of estimator directly affects PML privacy. We next specialize this characterization to the Kalman-filter estimate. We first consider the noiseless case, i.e., $\Xi=0$, to isolate the privacy effect of the estimator itself before introducing artificial release noise as a privacy-preserving mechanism. In this case, any privacy loss arises solely from Kalman filter gain $K^{\mathrm{KF}}$. The following proposition characterizes the resulting $(\varepsilon,\delta)$-PML privacy guarantee.

\begin{secprop}
\label{prop:kf_eps_delta_generalF}
Suppose that Assumptions~\ref{asm:rank_F},~\ref{asm:2} and~\ref{asm:3} hold.  
Then, the estimator characterized by~\eqref{eq:estimator} and~\eqref{eq:noise} with gain $K = K^{\mathrm{KF}}$ and noise covariance $\Xi = 0$ preserves $(\varepsilon,\delta)$-PML privacy for the secret variable $\zeta_{\infty}$ if and only if the pair $(\varepsilon, \delta)$ satisfies
\begin{align}
\label{eq:prop_state}
2\varepsilon
\ge
\mathbf{F}^{-1}_{\chi^2_{\ell^{\mathrm{KF}}}}(1-\delta)
+\log\det(\Sigma_{\zeta\zeta})
 -\log\det(\Gamma_{\zeta\zeta}^{\mathrm{KF}}),
\end{align}
where
\begin{align*}
&\ell^{\mathrm{KF}}
:=
\rank\bigl(F(\Sigma_\infty-P^{\mathrm{KF}}) H^{\top}\bigr), \\
&\Gamma_{\zeta\zeta}^{\mathrm{KF}}:=
F\Sigma_\infty F^\top
-
F(\Sigma_\infty-P^{\mathrm{KF}})H^\top \\
&\times \Bigl(H(\Sigma_\infty-P^{\mathrm{KF}})H^\top \Bigr)^\dagger
H(\Sigma_\infty-P^{\mathrm{KF}})F^\top,
\end{align*}
$\Sigma_{\infty}$ is the solution of~\eqref{eq:steady} and $P^{\mathrm{KF}}$ is defined in~\eqref{eq:PKF}.
\end{secprop}

\begin{proof}
See Appendix~\ref{app:prop3_4}.
\end{proof}

Proposition~\ref{prop:kf_eps_delta_generalF} provides an explicit characterization of the PML privacy achieved by the steady-state Kalman filter in the absence of release noise. The result shows that the privacy parameters are determined by the steady-state error covariance $P^{\mathrm{KF}}$. Hence, the estimation accuracy of the Kalman filter directly affects the amount of information that the released estimate may reveal about the secret variable.

 We next compare it with estimators obtained from other choices of the gain matrix $K$. This comparison is not immediate, because privacy is specified by the two-parameter pair $(\varepsilon,\delta)$ rather than by a single scalar quantity. To enable a meaningful comparison, we therefore fix $\delta \in (0,1)$ and define the privacy parameter associated with $K$ by
\begin{align}
\varepsilon(K, \delta)
:=&
\frac{1}{2}\mathbf{F}^{-1}_{\chi^2_\ell}(1-\delta)
 \nonumber\\&+\frac{1}{2}\log\det(\Sigma_{\zeta\zeta})
-\frac{1}{2}\log\det(\Gamma_{\zeta\zeta}). \label{eq:varepsilonK}
\end{align}
Then, an estimator of the form~\eqref{eq:estimator} with gain $K_1$ is said to be more private than another estimator with gain $K_2$ if
\begin{align*}
\varepsilon(K_1, \delta) \le \varepsilon(K_2, \delta).
\end{align*}

We are now ready to state the main privacy anti-optimality result for the Kalman filter. The following theorem shows that, under the structural condition $F=LH$ and in the absence of release noise, the Kalman filter is the least private estimator among stable linear time-invariant estimators of the form~\eqref{eq:estimator}.

\begin{secthm}
\label{thm:kalman_lp}
Suppose that Assumptions~\ref{asm:rank_F},~\ref{asm:2} and~\ref{asm:3} hold.  
Also assume that there exists a matrix $L \in \bR^{r \times s}$ such that $F = LH$
and that no release noise is added in~\eqref{eq:noise}, i.e., $\Xi =0$.
Then, the steady-state Kalman filter admits a finite value of $\varepsilon(K^{\mathrm{KF}},\delta)$ and is the least private estimator in the class of stable linear time-invariant estimators of the form~\eqref{eq:estimator}; that is,
\begin{align*}
\varepsilon(K, \delta) \le \varepsilon(K^{\mathrm{KF}}, \delta),
\qquad
\forall\, K\in\mathcal K, \delta \in (0,1),
\end{align*}
where $\mathcal K:=\{K\in\mathbb R^{n\times m}: A-KCA \text{ is Schur stable}\}.$
\end{secthm}

\begin{proof}
See Appendix~\ref{app:1}.
\end{proof}

Theorem~\ref{thm:kalman_lp} reveals a privacy anti-optimality property of the Kalman filter. Under the structural condition $F=LH$ and without release noise, the steady-state Kalman filter maximizes the required PML privacy budget $\varepsilon(K,\delta)$ for every fixed $\delta\in(0,1)$ among all stable linear estimators of the form~\eqref{eq:estimator}. Thus, although the Kalman filter is optimal for estimation, it is the least favorable estimator in this class from the PML-privacy perspective. In particular, the condition $F=LH$ holds in the full-release case $H=I_n$, showing that publishing the full Kalman state estimate can be highly undesirable for privacy.

\begin{secrem}
    The condition $F=LH$ is essential for
Theorem~\ref{thm:kalman_lp}. It ensures that the secret variable
$\zeta_k=Fx_k$ is a linear function of the released variable
$z_k=Hx_k$. Hence, a more accurate estimate of $z_k$ also yields a more
accurate estimate of $\zeta_k$, which allows the Kalman filter's
minimum-error property to be translated into maximal PML leakage. Without
this condition, estimation accuracy for $z_k$ and privacy leakage about
$\zeta_k$ need not be aligned, and the Kalman filter is not necessarily
the least private estimator.
\end{secrem}

% {In general, when $H \neq I$, the Kalman filter is not necessarily the least private estimator, since another filter may disclose more relevant information about $\zeta_{\infty}$. However, when $H = LF$ for some $L \in \bR^{s \times r}$, the above result continues to hold, and the proof follows in a similar manner.}

% {In general, the condition $\Cov(\hat{x}_{\infty}) \succ 0$ does not hold automatically. The following lemma provides a sufficient condition under which this property is guaranteed.}

% \begin{seclem}
% \label{lem:positive_cov}
% Assume that $A$ is Schur stable and $(A,B)$ is controllable. Let $\Sigma_\infty$ denote the unique steady-state covariance matrix satisfying
% \begin{align*}
%     \Sigma_\infty = A\Sigma_\infty A^\top + BB^\top.
% \end{align*}
% Then $\Sigma_\infty \succ 0$. Define the backward state-transition matrix by
% \begin{align*}
%     A_b := \Sigma_\infty A^\top \Sigma_\infty^{-1}.
% \end{align*}
% If, in addition, $(A_b,C)$ is observable, then the steady-state Kalman estimate satisfies
% \begin{align*}
%     \Cov(\hat{x}_\infty^{\mathrm{KF}}) \succ 0.
% \end{align*}
% \end{seclem}

% \begin{proof}
% See Appendix~\ref{app:2}
% \end{proof}

\section{Design of Privacy-preserving Estimators}
\label{sec:filter}
This section addresses  Problem~\ref{prob:1}, namely how to design the estimator gain and release noise covariance to satisfy a prescribed PML privacy constraint while maintaining estimation performance. However, we first establish a fundamental limitation: under the structural condition $F=LH$, any estimator satisfying $(\varepsilon,\delta)$-PML privacy must incur a nontrivial steady-state estimation error. We then turn to constructive privacy-preserving estimator designs, which do not rely on the condition $F = LH$. Specifically, we provide sufficient conditions on release noise covariance $\Xi$ in the case $K = K^{\mathrm{KF}}$. Then, we jointly design the estimator gain and release-noise covariance.

Before presenting the lower bound and the design results, we introduce a standing feasibility condition on the privacy parameters. 
By Corollary~3.4 of~\cite{liu2025privacy}, an $(\varepsilon,\delta)$-PML private release can exist only if
\begin{align}
\label{eq:privacy_feasibility_condition_ell}
\varepsilon
>
\frac{1}{2}\fF^{-1}_{\chi^2_\ell}(1-\delta),
\end{align}
where $\ell = \rank(\Sigma_{\zeta\hat z})$ is defined in Lemma~\ref{lem:gaussian_pml_characterization}. However, this rank generally depends on the estimator gain and $H$, and therefore is not fixed a priori. To obtain a uniform condition that is independent of the particular estimator, we impose the following stronger standing assumption.
\begin{secasm} 
\label{asm:privacy}
 The pair $(\varepsilon, \delta)$ satisfies the following inequality:
    \begin{align}
\label{eq:privacy_feasibility_condition}
\varepsilon
>
\frac{1}{2}\fF^{-1}_{\chi^2_r}(1-\delta),
\end{align}
where $r$ is the dimension of the secret variable $\zeta_k=Fx_k$.
\end{secasm}

Since $\ell\le r$ and $\fF^{-1}_{\chi^2_j}(1-\delta)$ is nondecreasing in $j$, condition~\eqref{eq:privacy_feasibility_condition} implies~\eqref{eq:privacy_feasibility_condition_ell} for any possible rank $\ell$. For notational convenience, for each $j\in\{0,1,\ldots,r\}$, define
\begin{align}
\label{eq:kappa}
\kappa_j(\varepsilon,\delta)
:=
\exp\left(
\frac{\fF^{-1}_{\chi^2_j}(1-\delta)-2\varepsilon}{r}
\right).
\end{align}
It follows directly from~\eqref{eq:privacy_feasibility_condition} that
\begin{align*}
0<\kappa_1(\varepsilon,\delta) < \kappa_2(\varepsilon,\delta)< \cdots < \kappa_r(\varepsilon,\delta) <1.
\end{align*}

With the feasibility condition and the quantities $\kappa_j(\varepsilon,\delta)$ defined above, we are ready to state the first result of this section, which quantifies the estimation-error penalty imposed by PML privacy under the structural condition $F=LH$.

\subsection{Estimation Lower Bound Induced by $(\varepsilon,\delta)$-PML Privacy}
The following proposition is derived from Lemma~\ref{lem:gaussian_pml_characterization}. 
It shows that, under the structural condition $F=LH$, PML privacy imposes a nontrivial lower bound on the steady-state estimation error.

\begin{secprop}
\label{prop:pml_error_lower_bound}
Suppose that Assumptions~\ref{asm:rank_F}--\ref{asm:2} hold. 
Assume that there exists an $L\in\bR^{r\times s}$ such that $F=LH$. 
Consider any stable estimator of the form~\eqref{eq:estimator}--\eqref{eq:noise} that preserves $(\varepsilon,\delta)$-PML privacy for $\zeta_\infty$, where the finite $\varepsilon \geq 0$ and $\delta \in (0,1)$ satisfy Assumption~\ref{asm:privacy}. 
Then,
\begin{align}
\label{eq:pml_error_lower_bound}
\tr\bigl(\Cov(z_\infty-\hat z_\infty)\bigr)
\ge
\frac{
r\det(\Sigma_{\zeta\zeta})^{1/r}\kappa_1(\varepsilon,\delta)
}{
\det(LL^\top)^{1/r}
}.
\end{align}
\end{secprop}

\begin{proof}
The proof is provided in Appendix~\ref{app:prop4_2}.
\end{proof}

The lower bound on the estimation error holds independently of the specific estimator gain $K$ and release noise covariance $\Xi$. Moreover, the bound depends explicitly on the privacy parameters through $\kappa_1(\varepsilon,\delta)$. A smaller privacy budget $\varepsilon$, or a smaller failure probability $\delta$, increases $\kappa_1(\varepsilon,\delta)$, thereby raising the lower bound and enforcing a larger estimation error.

\subsection{Noise Design for Privacy Preservation in the Kalman Filter}
Although the Kalman filter may not be ideal from a privacy perspective, it is widely used in practice. It is therefore important to design the noise covariance $\Xi$ to preserve privacy when the Kalman filter is employed. The following theorem provides a method to design $\Xi$ for the Kalman filter.
\begin{secthm}
\label{thm:noise_KF}
    Suppose that Assumptions~\ref{asm:rank_F}--\ref{asm:2} hold. For any finite $\varepsilon \geq 0$ and $\delta \in (0,1)$ satisfying Assumption~\ref{asm:privacy}, the estimator characterized by~\eqref{eq:estimator} and~\eqref{eq:noise} with gain $K = K^{\mathrm{KF}}$ and noise covariance $\Xi = \Xi^{\mathrm{KF}}$ preserves $(\varepsilon,\delta)$-PML privacy for the secret variable $\zeta_{\infty}$ if
    \begin{align}
        \Xi^{\mathrm{KF}} \succeq & \frac{\kappa_{\ell^{\mathrm{KF}}} (\varepsilon,\delta)}{1-\kappa_{\ell^{\mathrm{KF}}} (\varepsilon,\delta)} H (\Sigma_{\infty} -P^{\mathrm{KF}})\nonumber \\
        \times &F^{\top}\Sigma_{\zeta\zeta}^{-1}F (\Sigma_{\infty} -P^{\mathrm{KF}}) H^\top,  \label{eq:noise_XiKF}
    \end{align}
    where $\Sigma_{\infty}$ is defined in~\eqref{eq:steady},  $P^{\mathrm{KF}}$ is defined in~\eqref{eq:PKF} and $\ell^{\mathrm{KF}}$ is defined in Proposition~\ref{prop:kf_eps_delta_generalF}.
\end{secthm}
\begin{proof}
    The proof is provided in Appendix~\ref{app:2}.
\end{proof}

As indicated by the right-hand side of~\eqref{eq:noise_XiKF}, a smaller Kalman error covariance $P^{\mathrm{KF}}$ makes $\Sigma_\infty-P^{\mathrm{KF}}$ larger in the positive semidefinite sense, which increases the covariance between the released estimate and the secret variable. Consequently, a larger release noise covariance is required to satisfy the same $(\varepsilon,\delta)$-PML privacy. Since this approach fixes the gain at $K^{\mathrm{KF}}$ and enforces privacy only through output perturbation, the next subsection considers joint design of the estimator gain and release noise covariance.

\subsection{Convex Optimization Design for Stable Linear Estimators}
In general, directly designing $K$ and $\Xi$ is intractable, as it leads to a nonconvex optimization problem. Nevertheless, a tractable convex relaxation can be derived. To this end, the next theorem provides a sufficient lower bound on the noise covariance $\Xi$ for guaranteeing privacy.

\begin{secthm}
\label{thm:noise_lb}
Suppose that Assumptions~\ref{asm:rank_F}--\ref{asm:2} hold. Let
$ A_e := A - KCA$
and assume that $A_e$ is Schur stable. 
For any finite $\varepsilon \geq 0$ and $\delta \in (0,1)$ satisfying Assumption~\ref{asm:privacy}, the estimator characterized by~\eqref{eq:estimator} and~\eqref{eq:noise} with gain $K$ and noise covariance $\Xi$ preserves $(\varepsilon,\delta)$-PML privacy for the secret variable $\zeta_{\infty}$ if
\begin{align}
    \Xi \succeq &
\frac{\kappa_r (\varepsilon,\delta)}{1-\kappa_r (\varepsilon,\delta)}
\|\Sigma_{\zeta\zeta}^{-1}\|_2
\operatorname{tr}(
W_F \Sigma_\infty C^\top C\Sigma_\infty
) \nonumber \\
\label{eq:noise_lb}
&\times \sum_{i=0}^{\infty}\|H A_e^i K \|_2^2 I_s,
\end{align}
where $\Sigma_{\zeta \zeta}$ is defined in~\eqref{eq:steady_zeta} and $W_F$ is the solution of 
\begin{align*}
    W_F = A^\top W_F A + F^\top F.
\end{align*}
\end{secthm}
\begin{proof}
See Appendix~\ref{app:3}.
\end{proof}

    The sufficient covariance bound in~\eqref{eq:noise_lb} admits a natural interpretation. First, $W_F$ is the observability Gramian associated with the pair $(A,F)$. Hence,
$\operatorname{tr}(W_F \Sigma_\infty C^\top C \Sigma_\infty)$
quantifies how strongly the private variable $\zeta$ propagates through the system dynamics and the estimator. 
Second, the term $\sum_{i=0}^{\infty}\|H A_e^i K\|_2^2$ also determines the amount of release noise in~\eqref{eq:noise_lb}. In general, reducing this term conflicts with estimation performance objectives, since accurate estimation often requires smaller $\|A_e\|_2$ for faster decay rate. This may require a larger gain $K$, which can in turn increase $\|H A_e^i K\|_2^2$. Consequently, a larger noise covariance $\Xi$ is required according to~\eqref{eq:noise_lb}. This tradeoff motivates the joint design of $K$ and $\Xi$ in the sequel.

Motivated by Theorem~\ref{thm:noise_lb}, we now develop a tractable synthesis method for jointly designing the estimator gain and the release noise covariance. The following theorem formulates this co-design problem as a semi-definite programming (SDP) which is convex. Any feasible solution yields a stable linear time-invariant estimator that satisfies the prescribed $(\varepsilon,\delta)$-PML privacy constraint and admits an explicit upper bound on the steady-state estimation error.

\begin{secthm}
\label{thm:convex_opt}
Suppose that Assumptions~\ref{asm:rank_F}--\ref{asm:2} hold. For any finite $\varepsilon \geq 0$ and $\delta \in (0,1)$ satisfying Assumption~\ref{asm:privacy}, consider the following SDP:
\begin{align}
&\min_{\substack{
X\in \bS^{n}_{++},\; U \in \bR^{m \times n}\\
Z\in \bS^{p+q}_{++},\; W \in \bS^{m}_{++}
}}
\quad
\operatorname{tr}(Z)+\beta(\varepsilon,\delta)\operatorname{tr}(W)
\label{eq:SDPobj}\\
\text{s.t.}\quad &
% \begin{bmatrix}
% X & A^\top X+ A^\top C^\top U & H^\top\\
% X A+U^\top C A & X & 0\\
% H & 0 & I
% \end{bmatrix} \succ 0, \label{eq:LMI1}\\
% &\begin{bmatrix}
% Z &  \begin{bmatrix}
%     B^{\top}X + B^\top C^\top U\\
%     D^\top U
% \end{bmatrix}\\
% \begin{bmatrix}
%     XB + U^\top CB &
%      U^\top D
% \end{bmatrix} & X
% \end{bmatrix}\succeq 0, \label{eq:LMI2}\\
% &\begin{bmatrix}
% W & U\\
% U^\top & X
% \end{bmatrix}\succeq 0, \label{eq:LMI3}
\begin{bmatrix}
X & \star & \star\\
X A+U^\top C A & X & \star\\
H & 0 & I
\end{bmatrix} \succ 0, \label{eq:LMI1}\\
&\begin{bmatrix}
Z & \star\\
\begin{bmatrix}
    XB + U^\top CB &
     U^\top D
\end{bmatrix} & X
\end{bmatrix}\succeq 0, \label{eq:LMI2}\\
&\begin{bmatrix}
W & \star\\
U^\top & X
\end{bmatrix}\succeq 0, \label{eq:LMI3}
\end{align}
where $\star$ denotes entries implied by symmetry,
\begin{align*}
\beta(\varepsilon,\delta):=
\frac{\kappa_r (\varepsilon,\delta)\|\Sigma_{\zeta\zeta}^{-1}\|_2\, s}
{1-\kappa_r (\varepsilon,\delta)}
\operatorname{tr}\!\bigl(
W_F \Sigma_\infty C^\top C\Sigma_\infty
\bigr),
\end{align*}
and $s$ is the dimension of $z_k$. Assume that the SDP is feasible. Using a feasible solution $(X,U,Z,W)$
of \eqref{eq:SDPobj}--\eqref{eq:LMI3}, define
\begin{align*}
K=-X^{-1}U^\top, \Xi=\frac{\beta(\varepsilon,\delta)}{s}\operatorname{tr}(W)I_s.
\end{align*}
Then, the estimator characterized by~\eqref{eq:estimator} and~\eqref{eq:noise} with gain $K$ and noise covariance $\Xi$ preserves $(\varepsilon,\delta)$-PML privacy for the secret variable $\zeta_\infty$, and
\begin{align*}
\lim_{k\to\infty}\operatorname{tr}\bigl(\Cov(z_k-\hat z_k)\bigr)
\le
\operatorname{tr}(Z)+\beta(\varepsilon,\delta)\operatorname{tr}(W)
\end{align*}
holds. 
\end{secthm}
\begin{proof}
    See Appendix~\ref{app:4}.
\end{proof}

Theorem~\ref{thm:convex_opt} provides a tractable sufficient design for privacy-preserving estimation. The LMI~\eqref{eq:LMI1} enforces stability of the estimation error dynamics, \eqref{eq:LMI2} bounds the steady-state estimation error, and \eqref{eq:LMI3} bounds the privacy-relevant sensitivity term appearing in~\eqref{eq:noise_lb}. The resulting release-noise covariance is conservative because the designed noise covariance $\Xi$ is proportional to $I_s$. This choice of the noise covariance $\Xi$ ignores the directional structure of the information leakage. However, it enables the joint design of $K$ and $\Xi$ to be formulated as an SDP. Therefore, any feasible solution of the SDP yields a stable linear time-invariant estimator and a release noise covariance that satisfy the prescribed $(\varepsilon,\delta)$-PML privacy constraint.

The objective $\operatorname{tr}(Z)+\beta(\varepsilon,\delta)\operatorname{tr}(W)$ has a natural privacy--utility interpretation. The term $\operatorname{tr}(Z)$ bounds the estimation error, while $\operatorname{tr}(W)$ determines the amount of release noise required by the sufficient privacy condition. Stronger privacy requirements increase $\beta(\varepsilon,\delta)$, thereby placing more weight on $\operatorname{tr}(W)$.

The following remark shows that the SDP is not vacuous: under the standing assumptions, a feasible solution can always be constructed.

\begin{secrem}
In fact, the optimization problem in Theorem~\ref{thm:convex_opt} is always feasible under Assumptions~\ref{asm:rank_F}--\ref{asm:2}.  To construct a trivial feasible point, set $U = 0$.
Since $A$ is Schur stable, for any $Q\succ0$ the discrete-time Lyapunov equation
\begin{align*}
X=A^\top X A+H^\top H+Q
\end{align*}
admits a unique solution $X\succ0$. Hence,
\begin{align*}
X \succ A^\top X A+H^\top H,
\end{align*}
which is equivalent to~\eqref{eq:LMI1} with $U=0$. Next, \eqref{eq:LMI2} reduces to
\begin{align*}
\begin{bmatrix}
Z & \begin{bmatrix} B^\top X \\ 0 \end{bmatrix}\\
\begin{bmatrix} XB & 0 \end{bmatrix} & X
\end{bmatrix}\succeq 0,
\end{align*}
and this can always be satisfied, for example, by choosing
\begin{align*}
Z=\operatorname{diag}(B^\top X B+\epsilon I_p,\epsilon I_q)
\end{align*}
with sufficiently large $\epsilon>0$. Finally, \eqref{eq:LMI3} becomes
\begin{align*}
\begin{bmatrix}
W & 0\\
0 & X
\end{bmatrix}\succeq 0,
\end{align*}
which clearly holds for any $W\succ0$. Therefore, the SDP always admits at least one feasible solution.
\red
\end{secrem}

%%%%%%%%%%%%%%%%%%%%%%%%%%%%%%%%%%%%%%%%%%%%%%%%%%%%%%%%%%%%%%%%%%%%%%%
%%%%%%%%%%%%%%%%%%%%%%%%%%%%%%%%%%%%%%%%%%%%%%%%%%%%%%%%%%%%%%%%%%%%%%%
%%%%%%%%%%%%%%%%%%%%%%%%%%%%%%%%%%%%%%%%%%%%%%%%%%%%%%%%%%%%%%%%%%%%%%%

%%%%%%%%%%%%%%%%%%%%%%%%%%%%%%%%%%%%%%%%%%%%%%%%%%%%%%%%%%%%%%%%%%%%%%%
%%%%%%%%%%%%%%%%%%%%%%%%%%%%%%%%%%%%%%%%%%%%%%%%%%%%%%%%%%%%%%%%%%%%%%%
%%%%%%%%%%%%%%%%%%%%%%%%%%%%%%%%%%%%%%%%%%%%%%%%%%%%%%%%%%%%%%%%%%%%%%%

\section{Simulations}
\label{sec:sim} 
In this section, we present two simulations. In the first subsection, we verify Theorem~\ref{thm:kalman_lp} using a simple numerical example. In the second subsection, we revisit the motivation example to demonstrate the effectiveness of the design methods developed in Section~\ref{sec:filter}.
\subsection{Numerical Verification of Theorem~\ref{thm:kalman_lp}}
\begin{figure}[t]
    \centering    \includegraphics[width=1\linewidth]{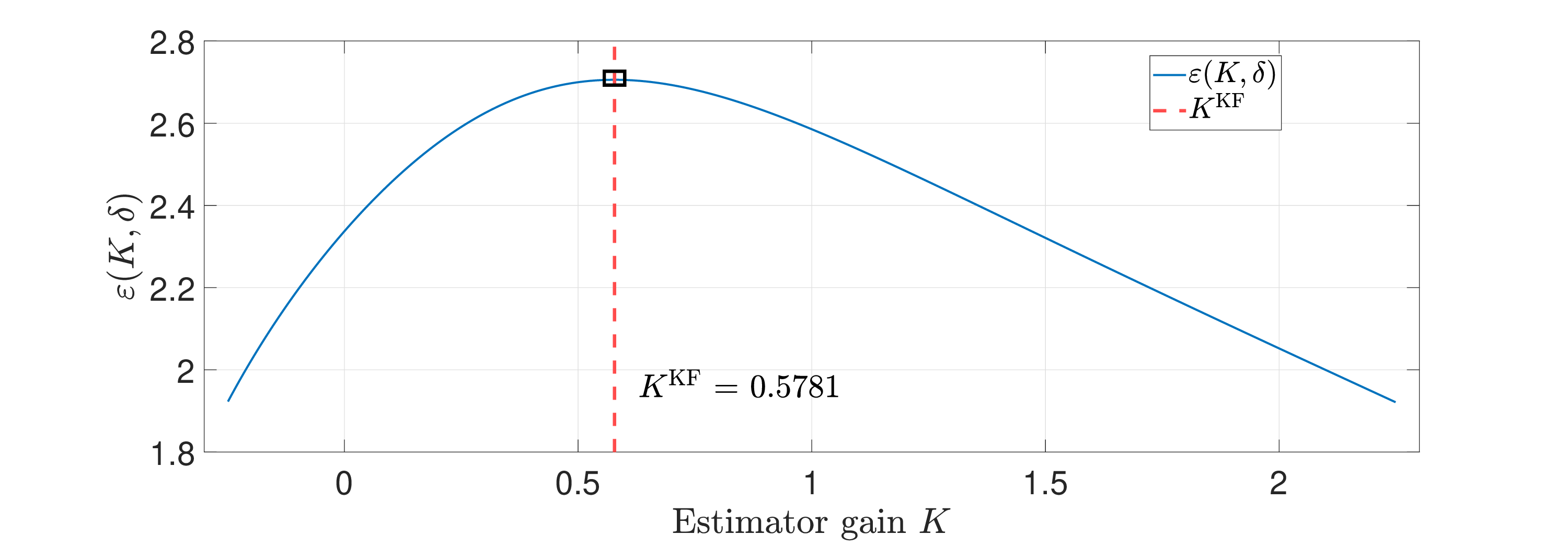}
    \caption{Privacy level $\varepsilon(K,\delta)$ versus estimator gain
    $K$ for the scalar example, with fixed $\delta=0.05$. The dashed
    vertical line indicates the steady-state Kalman gain
    $K^{\mathrm{KF}}$. The maximum is attained at the Kalman gain, in
    agreement with Theorem~\ref{thm:kalman_lp}.}
    \label{fig:eps_vs_K}
\end{figure}
To directly illustrate Theorem~\ref{thm:kalman_lp}, we consider a simple
one-dimensional example. The system is given by setting $
A=0.8, B=1, C=1, D=1$ in~\eqref{eq:LTI}. 
The released variable and the secret variable are the state itself, namely, $H =1$ and $F = 1$.
Hence, the structural condition in Theorem~\ref{thm:kalman_lp} is
satisfied with $L=1$.

Under the condition that no release noise is added, i.e., $\Xi = 0$, we fix $\delta=0.05$ and vary the estimator gain $K$ over the set of
stable gains satisfying $
|A-KCA|<1 $.
For each admissible $K$, the privacy parameter
$\varepsilon(K,\delta)$ is computed from the characterization in~\eqref{eq:varepsilonK}. The privacy metric can be evaluated directly from the steady-state
joint covariance of $(x_k,\hat x_k)$. Figure~\ref{fig:eps_vs_K} plots
$\varepsilon(K,\delta)$ as a function of $K$, together with a vertical
line marking the steady-state Kalman gain $K^{\mathrm{KF}}$.

As predicted by Theorem~\ref{thm:kalman_lp}, the maximum of
$\varepsilon(K,\delta)$ is attained at $K=K^{\mathrm{KF}}$. Therefore,
among all stable scalar estimators of the form~\eqref{eq:estimator}, the
Kalman filter is the least private one in the sense of PML.

\subsection{Revisiting the Motivation Example}

\begin{figure*}[t]
    \centering   \includegraphics[height=0.4\textheight,keepaspectratio]{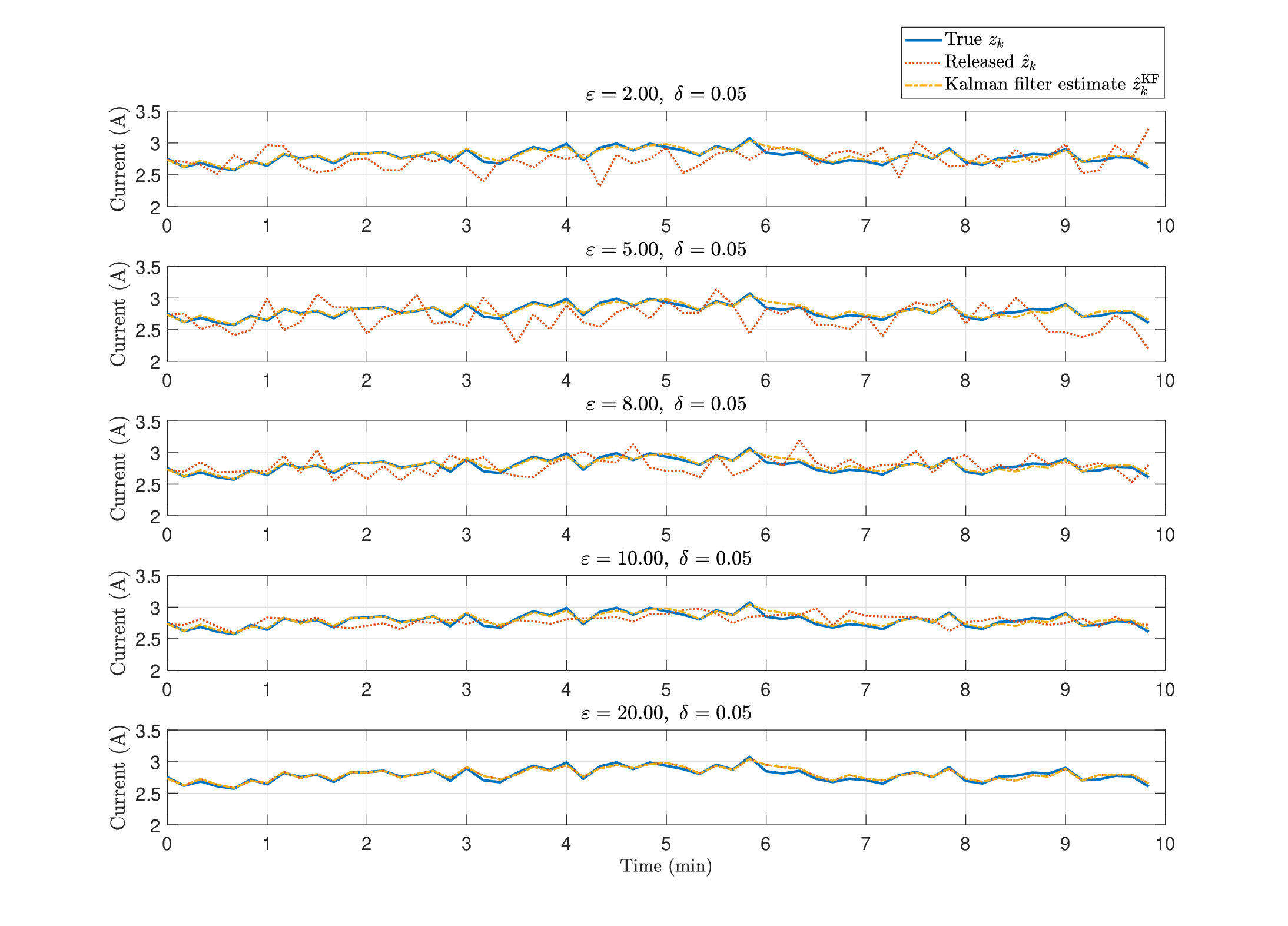}
    \caption{Trajectory comparison of the released estimates under different privacy parameters.}
    \label{fig:traj_compare}
\end{figure*}
\begin{figure*}
    \centering    \includegraphics[width=0.8\linewidth]{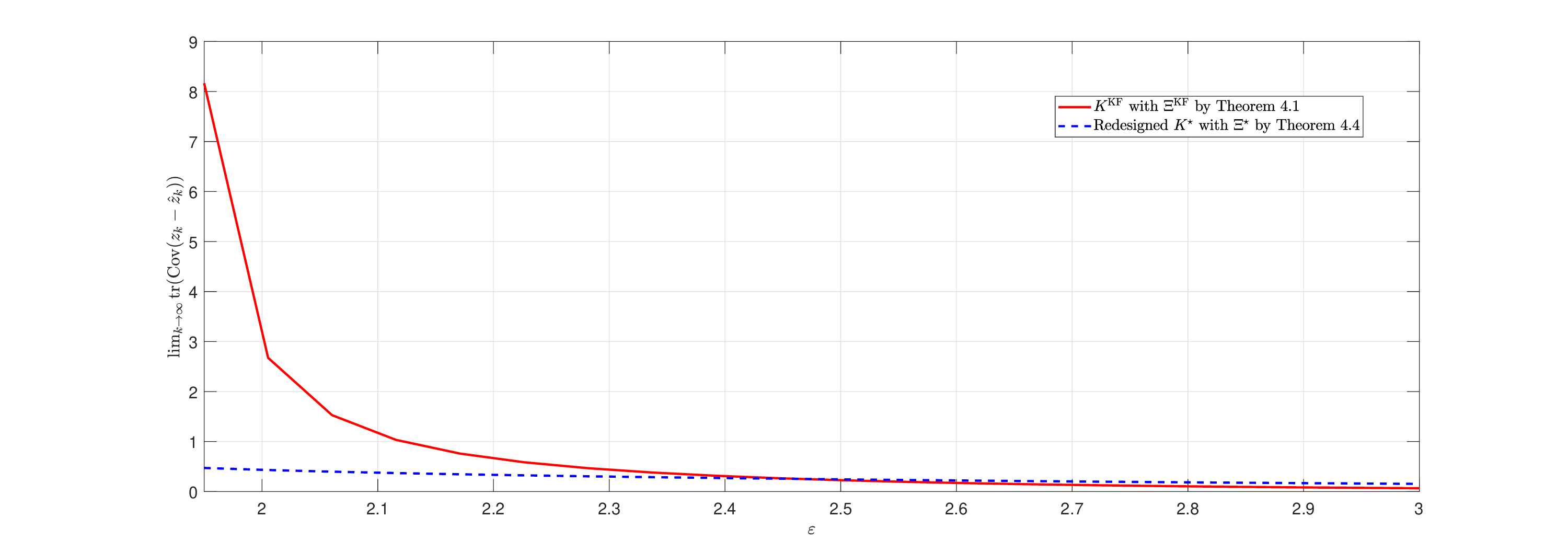}
    \caption{Theoretical steady-state estimation error comparison with fixed $\delta=0.05$.}
    \label{fig:KFvsK}
\end{figure*}
In this subsection, we illustrate the proposed privacy-preserving filter design using the source--load power system introduced in the motivating example. The secret variable is the user consumption state $d$, whereas the released estimate corresponds to the line current $i_s$. Thus, $F =
\begin{bmatrix}
0 & 0 & 1 & 0
\end{bmatrix},
\ 
H =
\begin{bmatrix}
\frac{1}{R} & -\frac{1}{R} & 0 & 0
\end{bmatrix}. $
For this source--load example, the structural condition $F=LH$ used in Theorem~\ref{thm:kalman_lp} is not satisfied. Therefore, this simulation is intended to illustrate the general SDP-based co-design method in Theorem~\ref{thm:convex_opt}, which does not require $F=LH$.
The physical parameters are chosen as $R = 0.5, R_L = 80, C_1 = 5, C_2 = 2$,
and the PI voltage-control gains are taken as $k_p = 1.0, k_i = 0.8$.
For the stochastic load model, we use $
\lambda = 0.0001, b_1 = 0.10, b_2 = 0.10, b_3 = 0.01$,
while the measurement noise standard deviation is set to
$ \sigma_v = 0.05$. The choice of $\lambda$ approximates a constant-load consumption.

% The continuous-time source--load model is first transformed into equilibrium coordinates and then discretized. Specifically, we shift the system around the deterministic equilibrium $v_s^\ast = v_{\mathrm{ref}},
% v_\ell^\ast = \frac{R_L}{R+R_L}v_{\mathrm{ref}},
% d^\ast = 0,
% \eta^\ast = -\frac{v_{\mathrm{ref}}}{k_i(R+R_L)}$.
% and define the state by $x(t)=
% \begin{bmatrix}
% v_s(t)-v_s^\ast &
% v_\ell(t)-v_\ell^\ast &
% d(t)-d^\ast &
% \eta(t)-\eta^\ast
% \end{bmatrix}^{\top}$.
% This yields a linear stochastic system of the form
% \begin{align*}
% \mathrm{d}x(t)=A_cx(t)\,\mathrm{d}t+G_c\mathrm{d}W(t).
% \end{align*}
The model is discretized with sampling period $T_s=10$ using exact discretization, leading to
\begin{align*}
x_{k+1}=Ax_k + Bw_k,
\end{align*}
with
$
A=\exp(A_cT_s),\qquad
BB^\top=\int_0^{T_s}\exp(A_c\tau)\,G_cG_c^\top \exp(A_c^\top\tau)\,\mathrm{d}\tau.
$

 In the simulation, the privacy confidence level is fixed at 
$\delta = 0.05$
and several privacy levels are considered by varying $\varepsilon \in \{2,\,5,\,8,\,10,\,20\}$.
For each value of $\varepsilon$, the corresponding optimal gain $K^\star$ and privacy noise covariance $\Xi^\star$ are computed from the SDP. To ensure a fair comparison across privacy levels, all estimators are evaluated on the same realization of the system trajectory and measurement sequence.
As a benchmark, we also compute the steady-state Kalman filter, which corresponds to the minimum mean-square error estimator without privacy protection.

Figure~\ref{fig:traj_compare} shows the estimation trajectories for different values of $\varepsilon$, together with the true line current and the steady-state Kalman filter estimate with $\Xi^{\mathrm{KF}} = 0$. Each subfigure corresponds to one privacy level $\varepsilon$. As expected, when the privacy requirement becomes more stringent, that is, when $\varepsilon$ decreases, the released estimate has larger estimation error. In contrast, for larger $\varepsilon$, the released estimate becomes closer to the Kalman filter benchmark. This explicitly shows the tradeoff between the estimation performance and the privacy.

Figure~\ref{fig:KFvsK} compares the estimation performance of the Kalman filter with noise design in Theorem~\ref{thm:noise_KF} with that of the SDP-based design in Theorem~\ref{thm:convex_opt} for different values of $\varepsilon$. For small values of $\varepsilon$, the noise design for the Kalman filter requires a relatively large release noise covariance $\Xi^{\mathrm{KF}}$, resulting in a noticeably larger estimation error. As $\varepsilon$ increases, the two designs exhibit similar performance. The performance improvement of the SDP-based design can be understood from the theoretical bounds developed in Section~\ref{sec:filter}. In the Kalman filter noise design, the gain is fixed as $K=K^{\mathrm{KF}}$, and privacy is enforced only by increasing the release noise covariance $\Xi^{\mathrm{KF}}$. In contrast, the proposed SDP jointly designs the estimator gain $K$ and the noise covariance $\Xi$. This additional degree of freedom allows the estimator to reduce the privacy-relevant sensitivity term in~\eqref{eq:noise_lb}, thereby requiring a smaller release noise covariance to achieve the same $(\varepsilon,\delta)$-PML privacy level. Therefore, the advantage of the proposed SDP-based design is most pronounced when stronger privacy protection is required.

%%%%%%%%%%%%%%%%%%%%%%%%%%%%%%%%%%%%%%%%%%%%%%%%%%%%%%%%%%%%%%%%%%%%%%%
%%%%%%%%%%%%%%%%%%%%%%%%%%%%%%%%%%%%%%%%%%%%%%%%%%%%%%%%%%%%%%%%%%%%%%%
%%%%%%%%%%%%%%%%%%%%%%%%%%%%%%%%%%%%%%%%%%%%%%%%%%%%%%%%%%%%%%%%%%%%%%%

\section{Conclusion}
\label{sec:con}

This paper has studied privacy-preserving estimation for linear stochastic systems under steady-state PML privacy. A linear function of the state has been treated as the secret variable, while another linear function has been estimated and released through a stable linear estimator with additive Gaussian perturbation. 

We have first analyzed the Kalman filter from the PML-privacy perspective. Although the Kalman filter is optimal in the minimum mean-square error sense, we have shown that, under a structural condition and in the absence of release noise, it is the least private estimator among the considered class of stable linear time-invariant estimators. Under the same structural condition, we have further established that $(\varepsilon,\delta)$-PML privacy imposes a nontrivial lower bound on the steady-state estimation error.

We have then derived sufficient release-noise covariance conditions for guaranteeing $(\varepsilon,\delta)$-PML privacy, both for the Kalman filter and for general stable linear estimators. Finally, we have proposed a convex SDP for jointly designing the estimator gain and release noise covariance. Any feasible solution guarantees the prescribed PML privacy level and provides an explicit upper bound on the steady-state estimation error. Specifically, the SDP-based design improves over the Kalman-filter baseline by jointly tuning the estimator and the release noise, which reduces unnecessary utility loss while satisfying the required privacy constraint. Simulations on a source--load power system have illustrated the resulting privacy--utility tradeoff and the advantage of the SDP-based design.

%%%%%%%%%%%%%%%%%%%%%%%%%%%%%%%%%%%%%%%%%%%%%%%%%%%%%%%%%%%%%%%%%%%%%%%
%%%%%%%%%%%%%%%%%%%%%%%%%%%%%%%%%%%%%%%%%%%%%%%%%%%%%%%%%%%%%%%%%%%%%%%
%%%%%%%%%%%%%%%%%%%%%%%%%%%%%%%%%%%%%%%%%%%%%%%%%%%%%%%%%%%%%%%%%%%%%%%

%%%%%%%%%%%%%%%%%%%%%%%%%%%%%%%%%%%%%%%%%%%%%%%%%%%%%%%%%%%%%%%%%%%%%%%
%%%%%%%%%%%%%%%%%%%%%%%%%%%%%%%%%%%%%%%%%%%%%%%%%%%%%%%%%%%%%%%%%%%%%%%
%%%%%%%%%%%%%%%%%%%%%%%%%%%%%%%%%%%%%%%%%%%%%%%%%%%%%%%%%%%%%%%%%%%%%%%

\appendices
\section{Proof of Lemma~\ref{lem:gaussian_pml_characterization}}
\label{app:lem1}
We first prove sufficiency. Suppose that $\Gamma_{\zeta\zeta}\succ0$ and
that \eqref{eq:lem_privacy} holds. Then the prior and posterior
distributions of $\zeta_\infty$ are both nondegenerate Gaussian
distributions. Hence, Theorem~3.3 in~\cite{liu2025privacy} applies
directly and yields
\begin{align*}
\mathbb P\bigl(\ell(\zeta_\infty\to \hat z_\infty)\le \varepsilon\bigr)
\ge 1-\delta .
\end{align*}
Therefore, the estimator preserves $(\varepsilon,\delta)$-PML privacy.

We next prove necessity. Suppose that the estimator preserves
$(\varepsilon,\delta)$-PML privacy for some finite $\varepsilon \ge 0$ and
$\delta \in (0,1)$. {We first show that this necessarily implies
$\Gamma_{\zeta\zeta}\succ0$ by contradiction.} Assume that $\Gamma_{\zeta\zeta}$ is singular. Then, there exists $0\neq a\in\operatorname{Ker}(\Gamma_{\zeta\zeta})$. Let
$\tilde \zeta_{\hat z}$ be a random vector distributed according to the
conditional distribution of $\zeta_\infty$ given $\hat z_\infty=\hat z$,
i.e.,
\begin{align*}
\tilde \zeta_{\hat z}
\sim
\mathcal L(\zeta_\infty\mid \hat z_\infty=\hat z).
\end{align*}
Let $\mu(\hat z):=\mathbb E[\tilde \zeta_{\hat z}]$.
Since the covariance of $\tilde \zeta_{\hat z}$ is
$\Gamma_{\zeta\zeta}$, we have
\begin{align*}
\Cov(a^\top \tilde \zeta_{\hat z})
=
a^\top \Gamma_{\zeta\zeta}a
=
0.
\end{align*}
Since a scalar random variable with zero variance is almost surely
constant, it follows that
\begin{align*}
a^\top \tilde \zeta_{\hat z}
=
\mathbb E[a^\top \tilde \zeta_{\hat z}]
=
a^\top \mu(\hat z),
\qquad \text{a.s.}
\end{align*}
Therefore, the conditional distribution assigns probability one to the
hyperplane $S(\hat z):=
\{\zeta\in\mathbb R^r:
a^\top\zeta=a^\top \mu(\hat z)\}$,
that is,
\begin{align*}
\mathbb P_{\zeta_\infty\mid \hat z_\infty=\hat z}(S(\hat z))=1.
\end{align*}

Recall that, under
Assumptions~\ref{asm:rank_F} and~\ref{asm:2}, the prior covariance
satisfies $\Sigma_{\zeta\zeta}\succ 0$. By $\Sigma_{\zeta\zeta}\succ 0$ and $a\neq 0$, we have
$a^\top\Sigma_{\zeta\zeta}a>0$.Therefore, the scalar random
variable $a^\top\zeta_\infty$ is Gaussian with variance
$a^\top\Sigma_{\zeta\zeta}a>0$, and hence is nondegenerate. Consequently, we have
\begin{align*}
\mathbb P_{\zeta_\infty}(S(\hat z))
=
\mathbb P(a^\top\zeta_\infty=a^\top \mu(\hat z))
=
0.
\end{align*}

Therefore, for almost every realization
$\hat z$, the conditional distribution
$\mathbb P_{\zeta_\infty\mid \hat z_\infty=\hat z}$ is supported on
$S(\hat z)$, whereas the prior distribution $\mathbb P_{\zeta_\infty}$ assigns zero probability to
$S(\hat z)$. Hence,
$\mathbb P_{\zeta_\infty\mid \hat z_\infty=\hat z}$ is not absolutely
continuous with respect to $\mathbb P_{\zeta_\infty}$, and the density
ratio in~\eqref{eq:PML} cannot be essentially bounded by any finite
constant. Consequently, it holds that
\begin{align*}
\ell(\zeta_\infty\to \hat z) > \varepsilon
\end{align*}
for any finite $\varepsilon \geq 0$ and almost every realization $\hat z$. It follows that
\begin{align*}
\mathbb P\bigl(
\ell(\zeta_\infty\to \hat z_\infty)\le \varepsilon
\bigr)
=
0
<
1-\delta,
\end{align*}
which contradicts $(\varepsilon,\delta)$-PML privacy since
$\delta\in(0,1)$. Hence, we have  $\Gamma_{\zeta\zeta}\succ0$. 

 Having established that $\Gamma_{\zeta\zeta}\succ0$, we now complete the
necessity argument by deriving \eqref{eq:lem_privacy}. By the fact that
$\Sigma_{\zeta\zeta}\succ0$ and $\Gamma_{\zeta\zeta}\succ0$,  the prior and
posterior distributions of $\zeta_\infty$ are nondegenerate Gaussian
distributions. Thus, Theorem~3.3 in~\cite{liu2025privacy} applies, and the assumed
$(\varepsilon,\delta)$-PML privacy of the estimator yields
\eqref{eq:lem_privacy}. This proves the necessity of \eqref{eq:lem_privacy} and completes the
proof of the lemma. \QED

\section{Proof of Proposition~\ref{prop:kf_eps_delta_generalF}}
\label{app:prop3_4}
First, we show that $\Gamma_{\zeta\zeta}^{\mathrm{KF}}$ is positive definite. By the definition of  $\Gamma_{\zeta\zeta}^{\mathrm{KF}}$, we have \begin{align*}
    \Gamma_{\zeta\zeta}^{\mathrm{KF}} 
= & FP^{\mathrm{KF}} F^\top + F\Bigl((\Sigma_\infty-P^{\mathrm{KF}}) - (\Sigma_\infty-P^{\mathrm{KF}})H^\top \\
&\times \big(H(\Sigma_\infty-P^{\mathrm{KF}})H^\top \big)^\dagger
H(\Sigma_\infty-P^{\mathrm{KF}}) \Bigr)F^\top.
\end{align*} 
Since the second term is a generalized Schur complement, it is positive semidefinite. Hence, it holds that $\Gamma_{\zeta\zeta}^{\mathrm{KF}}
\succeq
FP^{\mathrm{KF}}F^\top$.
This, together with Assumption~\ref{asm:rank_F} and Lemma~\ref{lem:prior_pd}, yields that 
\begin{align}
\label{eq:positive_Gamma}
    \Gamma_{\zeta\zeta}^{\mathrm{KF}}   \succeq FP^{\mathrm{KF}} F^\top \succ 0.
\end{align}

    By Lemma~\ref{lem:gaussian_pml_characterization}, it suffices to identify
$\Gamma_{\zeta\zeta}^{\mathrm{KF}}
= \Cov(\zeta_{\infty} \mid \hat{z}_{\infty}^{\mathrm{KF}})$ and $\rank\bigl(\Cov(\zeta_{\infty},\hat z_\infty^{\mathrm{KF}})\bigr)$.
The Gaussian conditioning formula yields
\begin{align}
\label{eq:prop_1}
\Cov(\zeta_\infty\mid \hat z_\infty^{\mathrm{KF}})
&=
\Cov(\zeta_\infty)-\Cov(\zeta_\infty,\hat z_\infty^{\mathrm{KF}}) \nonumber\\
& \times \Cov(\hat z_\infty^{\mathrm{KF}})^\dagger
\Cov(\hat z_\infty^{\mathrm{KF}},\zeta_\infty)
\end{align}
For the Kalman filter, define $e_\infty^{\mathrm{KF}}:=x_\infty-\hat x_\infty^{\mathrm{KF}}$. Then
\begin{align*}
x_\infty=\hat x_\infty^{\mathrm{KF}}+e_\infty^{\mathrm{KF}},
\end{align*}
and the orthogonality principle~\cite{anderson2005optimal} yields
\begin{align*}
\Cov(e_\infty^{\mathrm{KF}},\hat x_\infty^{\mathrm{KF}})=0.
\end{align*}
Hence, we have
\begin{align*}
\Sigma_\infty
&=
\Cov(x_\infty) = \Cov(\hat x_\infty^{\mathrm{KF}} + e_\infty^{\mathrm{KF}}) \\
&=
\Cov(\hat x_\infty^{\mathrm{KF}})+\Cov(e_\infty^{\mathrm{KF}}).
\end{align*}
Since $P^{\mathrm{KF}}=\Cov(e_\infty^{\mathrm{KF}})$, it follows that
\begin{align*}
\Cov(x_\infty,\hat x_\infty^{\mathrm{KF}})
=
\Cov(\hat x_\infty^{\mathrm{KF}})
=
\Sigma_\infty-P^{\mathrm{KF}}.
\end{align*}
Therefore, this implies that $
    \Cov(\zeta_\infty,\hat z_\infty^{\mathrm{KF}}) = F(\Sigma_\infty-P^{\mathrm{KF}}) H^{\top} $ and $\Cov(\hat{z}_{\infty}^{\mathrm{KF}}) = H(\Sigma_\infty-P^{\mathrm{KF}})H^{\top}$. 
Substituting these covariance identities into~\eqref{eq:prop_1} gives
$\Gamma_{\zeta\zeta}^{\mathrm{KF}}$ as stated, and the expression for
$\ell^{\mathrm{KF}}$ follows from $ \Cov(\zeta_\infty,\hat z_\infty^{\mathrm{KF}}) = F(\Sigma_\infty-P^{\mathrm{KF}}) H^{\top}$.
\QED
\section{Proof of Theorem~\ref{thm:kalman_lp}}
\label{app:1}
Since our analysis is carried out in steady state, we assume that the system state in~\eqref{eq:LTI} and the estimate in~\eqref{eq:estimator} are initialized in their steady-state distribution at time $k=0$. Then, it suffices to prove the result for arbitrary $k \ge 0$. 

We divide the proof into three parts:
(i) We show that $\Gamma_{\zeta \zeta}^{\mathrm{KF}} = \Cov(\zeta_k \mid \hat{z}_k^{\mathrm{KF}})$ is minimal over all $\Gamma_{\zeta \zeta}$ generated by $K \in \mathcal{K}$.
(ii) We prove that $\mathbf{F}_{\ell^{\mathrm{KF}}}^{-1}(1-\delta) \ge \mathbf{F}_{\ell}^{-1}(1-\delta)$ for every $\ell = \operatorname{rank}(\Cov(\zeta_k,\hat{z}))$ induced by any $K \in \mathcal{K}$.
(iii) We combine the results in (i) and (ii) to prove the theorem.

% (i) Under the linear Gaussian model in~\eqref{eq:LTI} and \eqref{eq:estimator}, the pair $(x_k,\hat{x}_k)$ is jointly Gaussian for every gain $K \in \mathcal{K}$. It follows that the conditional mean of $x_k$ given $\hat{z}_k$ depends linearly on $\hat{z}_k$, namely,
% \begin{align}
% \label{eq:barxK_def_no_sigma}
% \bar x_k
% :=
% \mathbb E[x_k\mid \hat z_k]
% =
% \Sigma_{x\hat z}(\Sigma_{\hat z \hat{z}} )^{\dagger} \hat z_k,
% \end{align}
% where $\Sigma_{x\hat z}:=\operatorname{Cov}(x_k,\hat z_k)$ and $\Sigma_{\hat z \hat{z}}:=\operatorname{Cov}(\hat z_k)$.

% Define the residual
% \begin{align*}
%     r_k:=x_k-\bar x_k.
% \end{align*}
% By the orthogonality principle of conditional expectation and the Gaussianity of the variables, $r_k$ is independent of $\hat z_k$~\cite[Theorem 3.1]{poor1994signal}. Consequently, the conditional covariance of $x_k$ given $\hat z_k$ is exactly the covariance of $ r_k$. This yields
%  \begin{align}
%      &\Cov(\zeta_k\mid \hat z_k) \nonumber \\
%      &= F\Cov(x_k - \bar{x}_k)F^{\top} = \Cov(\zeta_k - F\bar{x}_k). \label{eq:thm1_cov}
%  \end{align}

(i) Let $e_k^{\mathrm{KF}}:=x_k-\hat x_k^{\mathrm{KF}}$. By the
orthogonality principle of the Kalman filter\cite[Sec.~5.2]{anderson2005optimal} and the joint Gaussianity of
the variables, $e_k^{\mathrm{KF}}$ is independent of
$\hat x_k^{\mathrm{KF}}$. Therefore, $e_k^{\mathrm{KF}}$ is also independent of
$\hat z_k^{\mathrm{KF}}=H\hat x_k^{\mathrm{KF}}$. Moreover,
$\mathbb E[e_k^{\mathrm{KF}}]=0$. Since $F=LH$, we have
\begin{align*}
\zeta_k
=
Fx_k
=
LH(\hat x_k^{\mathrm{KF}}+e_k^{\mathrm{KF}})
=
L\hat z_k^{\mathrm{KF}}+LH e_k^{\mathrm{KF}} .
\end{align*}
Taking the conditional expectation with respect to
$\hat z_k^{\mathrm{KF}}$ gives
\begin{align*}
\mathbb E[\zeta_k\mid \hat z_k^{\mathrm{KF}}]
&=
L\hat z_k^{\mathrm{KF}}
+
LH\,\mathbb E[e_k^{\mathrm{KF}}\mid \hat z_k^{\mathrm{KF}}] \\
&=
L\hat z_k^{\mathrm{KF}} .
\end{align*}
Therefore, it follows that
\begin{align*}
\Cov(\zeta_k\mid \hat z_k^{\mathrm{KF}})
&=
\Cov\bigl(\zeta_k-\mathbb E[\zeta_k\mid \hat z_k^{\mathrm{KF}}]\bigr) \\
&=
\Cov(\zeta_k-L\hat z_k^{\mathrm{KF}}).
\end{align*}

Since $\hat z_k$ is generated by a stable linear estimator of the
form~\eqref{eq:estimator}, it is a linear function of the measurement
history. Under the joint Gaussian assumption,
$\mathbb E[\zeta_k\mid \hat z_k]$ is linear in $\hat z_k$, and hence is
also a linear estimate of $\zeta_k$ based on the same measurement history.
By the linear minimum-mean-square-error property of the Kalman filter, the
estimate $F\hat x_k^{\mathrm{KF}}$ has no larger error covariance than
this estimate. Therefore, we have
\begin{align*}
\Cov(\zeta_k-\mathbb E[\zeta_k\mid \hat z_k])
\succeq
\Cov(\zeta_k-F\hat x_k^{\mathrm{KF}}).
\end{align*}
Since $F=LH$ and $\hat z_k^{\mathrm{KF}}=H\hat x_k^{\mathrm{KF}}$, we have
$F\hat x_k^{\mathrm{KF}}=L\hat z_k^{\mathrm{KF}}$. Thus, it follows that $\Cov(\zeta_k-\mathbb E[\zeta_k\mid \hat z_k])
\succeq
\Cov(\zeta_k-L\hat z_k^{\mathrm{KF}})$.
By noticing that $\Cov(\zeta_k-\mathbb E[\zeta_k\mid \hat z_k]) = \Cov (\zeta_k \mid \hat z_k)$, we have
\begin{align*}
\Cov(\zeta_k\mid \hat z_k)
\succeq
\Cov(\zeta_k\mid \hat z_k^{\mathrm{KF}}).
\end{align*}

(ii) Let $\ell^{\mathrm{KF}} := \operatorname{rank}(\Cov(\zeta_k, \hat{z}^{\mathrm{KF}}))$. Since 
$\mathbf{F}^{-1}_{\chi^2_{\ell}}$ is monotone increasing with respect to $\ell$, it suffices to prove $\ell^{\mathrm{KF}} \geq \ell$. 

By the orthogonality principle for the Kalman filter again, one has
\begin{align*}
\mathbb{E}[(x_k - \hat{x}_k^{\mathrm{KF}})\hat{x}_k^\top] = 0,
\end{align*}
which implies that $\mathbb{E}[(z_k - \hat{z}_k^{\mathrm{KF}})\hat{z}_k^\top] = 0$. This yields
\begin{align}
&\Sigma_{z\hat{z}} := \Cov(z_k, \hat{z}_k) \nonumber\\
\label{eq:thm1_3}
&= \mathbb{E}[(\hat{z}_k^{\mathrm{KF}} + z_k - \hat{z}_k^{\mathrm{KF}})\hat{z}_k^\top]= \mathbb{E}[\hat{z}_k^{\mathrm{KF}} \hat{z}_k^\top].
\end{align}
Similarly, we have
\begin{align*}
\Sigma_{z\hat{z}^{\mathrm{KF}}} &:= \Cov(z_k, \hat{z}_k^{\mathrm{KF}}) = \mathbb{E}[\hat{z}_k^{\mathrm{KF}}(\hat{z}_k^{\mathrm{KF}})^\top].
\end{align*}
Now let $a \in \ker(\Sigma_{z\hat{z}^{\mathrm{KF}}})$. Then, it can be shown that
\begin{align*}
a^\top \mathbb{E}[\hat{z}_k^{\mathrm{KF}}(\hat{z}_k^{\mathrm{KF}})^\top] a
= \mathbb{E}[(a^\top \hat{z}_k^{\mathrm{KF}})^2] = 0.
\end{align*}
Since $(a^\top \hat{z}_k^{\mathrm{KF}})^2 \ge 0$ almost surely, it follows that
\begin{align*}
a^\top \hat{z}_k^{\mathrm{KF}} = 0
\qquad \text{a.s.}
\end{align*}
Invoking~\eqref{eq:thm1_3}, we have $a^\top \Sigma_{z\hat z}= 0$,
which shows that
\begin{align*}
\ker(\Sigma_{z\hat z^{\mathrm{KF}}}) \subseteq \ker(\Sigma_{z\hat z}^{\top}).
\end{align*}
Using $\operatorname{Im}(\Sigma_{z\hat z})=\ker(\Sigma_{z\hat z}^\top)^\perp$ and the symmetry of
$\Sigma_{z\hat z^{\mathrm{KF}}}$, we obtain
\begin{align}
&\operatorname{Im}(\Sigma_{z\hat z})
=
\ker(\Sigma_{z\hat z}^{\top})^\perp \nonumber\\
&\subseteq
\ker(\Sigma_{z\hat z^{\mathrm{KF}}})^\perp 
= \ker(\Sigma^{\top}_{z\hat z^{\mathrm{KF}}})^\perp =
\operatorname{Im}(\Sigma_{z\hat z^{\mathrm{KF}}}). \label{eq:thm1_4}
\end{align}
Applying the linear map $L$ to the inclusion in~\eqref{eq:thm1_4} gives $\operatorname{Im}(L\Sigma_{z\hat z})
\subseteq
\operatorname{Im}(L\Sigma_{z\hat z^{\mathrm{KF}}})$,
and hence
\begin{align*}
\rank(L\Sigma_{z\hat z})
\le
\rank(L\Sigma_{z\hat z^{\mathrm{KF}}}).
\end{align*}
This implies that $\ell_{\mathrm{KF}} \geq \ell$ by noticing that $\Cov(\zeta_k, \hat{z}_k) = L \Sigma_{z \hat{z}}$, which completes the proof of this part.

(iii) Recall that
\begin{align*}
\varepsilon(K,\delta) =
\frac{1}{2}\mathbf{F}^{-1}_{\chi^2_\ell}(1-\delta)
+\log\det(\Sigma_{\zeta\zeta})
-\log\det(\Gamma_{\zeta\zeta}).
\end{align*}
By~\eqref{eq:positive_Gamma}, we have $\Gamma_{\zeta\zeta}^{\mathrm{KF}}\succ 0$.
This, together with part (i), shows that $ 0 \prec \Gamma_{\zeta\zeta}^{\mathrm{KF}} \preceq \Gamma_{\zeta\zeta}$. Therefore, it follows that 
\begin{align*}
\log\det(\Gamma_{\zeta\zeta}^{\mathrm{KF}})
\le
\log\det(\Gamma_{\zeta\zeta}).
\end{align*}
By part (ii), it holds that 
\begin{align*}
\mathbf{F}^{-1}_{\chi^2_{\ell}}(1-\delta) \leq \mathbf{F}^{-1}_{\chi^2_{{\ell}^{\mathrm{KF}}}}(1-\delta).
\end{align*}
Combining these two inequalities yields the desired result, which completes the proof.
\QED

\section{Proof of Proposition~\ref{prop:pml_error_lower_bound}}
\label{app:prop4_2}
By Assumption~\ref{asm:rank_F} and the relation $F=LH$, the matrix $L$ has full row rank. Thus, $LL^\top\succ0$, and the right-hand side of~\eqref{eq:pml_error_lower_bound} is well defined.

If $\ell=0$, then $\Sigma_{\zeta\hat z}=0$. Since the variables are jointly Gaussian, $\zeta_\infty$ and $\hat z_\infty$ are independent, and hence
$\Gamma_{\zeta\zeta}=\Sigma_{\zeta\zeta}$. Therefore, it holds that 
\begin{align}
\label{eq:det_gamma_lower_bound_positive_rank}
\det(\Gamma_{\zeta\zeta})
\ge
\det(\Sigma_{\zeta\zeta})
\kappa_1(\varepsilon,\delta)^r.
\end{align}

If $\ell>0$, then Lemma~\ref{lem:gaussian_pml_characterization} gives
\begin{align*}
\fF^{-1}_{\chi^2_\ell}(1-\delta)
\le
2\varepsilon
-\log\det(\Sigma_{\zeta\zeta})
+\log\det(\Gamma_{\zeta\zeta}).
\end{align*}
Since $\ell\ge1$ and $\fF^{-1}_{\chi^2_j}(1-\delta)$ is non-decreasing in $j$, we have $\fF^{-1}_{\chi^2_1}(1-\delta)
\le
\fF^{-1}_{\chi^2_\ell}(1-\delta)$. Then, it follows that 
\begin{align*}
    \fF^{-1}_{\chi^2_1}(1-\delta)
\le
2\varepsilon
-\log\det(\Sigma_{\zeta\zeta})
+\log\det(\Gamma_{\zeta\zeta}).
\end{align*}
 This further yields that $\det(\Gamma_{\zeta\zeta})
\ge
\det(\Sigma_{\zeta\zeta})
\kappa_1(\varepsilon,\delta)^r$, which is equivalent to~\eqref{eq:det_gamma_lower_bound_positive_rank}.

Since $F=LH$, we have $\zeta_\infty=Lz_\infty$. Hence, $L\hat z_\infty$ is an estimator of $\zeta_\infty$ based on $\hat z_\infty$. By the minimum mean-square error property of conditional expectation,
\begin{align*}
\Gamma_{\zeta\zeta}
=
\Cov(\zeta_\infty\mid \hat z_\infty)
\preceq
\Cov(\zeta_\infty-L\hat z_\infty).
\end{align*}
Using $\zeta_\infty=Lz_\infty$, this yields
\begin{align}
\label{eq:gamma_LEL}
\Gamma_{\zeta\zeta}
\preceq
L\Cov(z_\infty-\hat z_\infty)L^\top.
\end{align}

Since $0\preceq L^\top(LL^\top)^{-1}L\preceq I_s$, it follows that
\begin{align*}
&\tr(\Cov(z_\infty-\hat z_\infty)) \\
&\ge
\tr\left(\Cov(z_\infty-\hat z_\infty)L^\top(LL^\top)^{-1}L\right) \\
&=
\tr\left((LL^\top)^{-1}L \Cov(z_\infty-\hat z_\infty) L^\top\right).
\end{align*}
Combining this with~\eqref{eq:gamma_LEL} gives
\begin{align}
\label{eq:trace_gamma_lower_bound}
\tr\left(\Cov(z_\infty-\hat z_\infty)\right)
\ge
\tr\left((LL^\top)^{-1}\Gamma_{\zeta\zeta}\right).
\end{align}

Finally, applying the arithmetic--geometric mean inequality~\cite[Problem~2.3.P10]{horn2012matrix}, we have 
\begin{align*}
\tr\left((LL^\top)^{-1}\Gamma_{\zeta\zeta}\right)
&=
\tr\left((LL^\top)^{-1/2}
\Gamma_{\zeta\zeta}
(LL^\top)^{-1/2}\right) \\
&\ge
r
\left(
\det\left(
(LL^\top)^{-1/2}
\Gamma_{\zeta\zeta}
(LL^\top)^{-1/2}
\right)
\right)^{1/r} \\
&=
r
\left(
\frac{\det(\Gamma_{\zeta\zeta})}{\det(LL^\top)}
\right)^{1/r}.
\end{align*}
Substituting~\eqref{eq:det_gamma_lower_bound_positive_rank} into the above inequality proves~\eqref{eq:pml_error_lower_bound}.
\QED

\section{Proof of Theorem~\ref{thm:noise_KF}}
\label{app:2}
To prove the theorem, we first present the following auxiliary lemma, which gives a sufficient condition for adding noise to preserve $(\varepsilon, \delta)$-PML.
{
\renewcommand{\thesecthm}{\ref{app:2}.1}
\begin{seclem}
\label{lem:noise_add}
For any finite $\varepsilon \geq 0$ and $\delta \in (0,1)$ satisfying Assumption~\ref{asm:privacy},
the estimator with gain $K$ and noise covariance $\Xi$ preserves $(\varepsilon, \delta)$-PML privacy for the secret variable $\zeta_\infty$ if the following LMI holds:
\begin{align*}
\begin{bmatrix}
(1-\kappa_{\ell}(\varepsilon, \delta))\Sigma_{\zeta\zeta} & \Sigma_{\zeta \hat x}H^\top\\
H \Sigma_{\hat x\zeta} & H \Sigma_{\hat x\hat x}H^\top+\Xi
\end{bmatrix}\succeq 0,
\end{align*}
where $\ell:=\rank(\Sigma_{\zeta\hat z})$ and $\kappa_{\ell}(\varepsilon,\delta)$ is defined in~\eqref{eq:kappa}.
\end{seclem}
}
\begin{proof}
    The proof follows directly from Theorem~3.6 in \cite{liu2025privacy}.
\end{proof}

\begin{proof}[{Proof of Theorem~\ref{thm:noise_KF}}]
By Lemma~\ref{lem:noise_add} and the Schur complement lemma, the released $\hat{z}^{\mathrm{KF}}_{\infty}$ can preserve $(\varepsilon,\delta)$-PML if
\begin{align}
& H \Sigma_{\hat x^{\mathrm{KF}}\hat x^{\mathrm{KF}}}H^\top+\Xi \nonumber \\
& \label{eq:thm2_1} \succeq
\frac{1}{1-\kappa_{\ell^{\mathrm{KF}}} (\varepsilon,\delta)}
H \Sigma_{\hat x^{\mathrm{KF}} \zeta}\Sigma_{\zeta\zeta}^{-1}\Sigma_{\zeta \hat x^{\mathrm{KF}}}H^\top,
\end{align}
{where~$\Sigma_{\hat{x}^{\mathrm{KF}}\zeta} = (\Sigma_{\infty} - P^{\mathrm{KF}}) F^{\top}$ as shown in the proof of Proposition~\ref{prop:kf_eps_delta_generalF}.}

{We show~\eqref{eq:noise_XiKF} implies~\eqref{eq:thm2_1}.} Since we have
\begin{align*}
    \begin{bmatrix}
\Sigma_{\zeta\zeta} & \Sigma_{\zeta \hat x^{\mathrm{KF}}}H^\top\\
H \Sigma_{\hat x^{\mathrm{KF}}\zeta} & H \Sigma_{\hat x^{\mathrm{KF}}\hat x^{\mathrm{KF}}}H^\top
\end{bmatrix}\succeq 0,
\end{align*}
the Schur complement lemma again implies that $H \Sigma_{\hat x^{\mathrm{KF}}\hat x^{\mathrm{KF}}}H^\top \succeq 
H \Sigma_{\hat x^{\mathrm{KF}}\zeta}\Sigma_{\zeta\zeta}^{-1}\Sigma_{\zeta \hat x^{\mathrm{KF}}}H^\top$.
{Thus,~\eqref{eq:noise_lb} implies~\eqref{eq:thm2_1}.}
% Substituting this into~\eqref{eq:thm2_1}, it suffices to prove
% \begin{align*}
% \Xi
% &\succeq
% \frac{\kappa_{\ell^{\mathrm{KF}}} (\varepsilon,\delta)}{1-\kappa_{\ell^{\mathrm{KF}}} (\varepsilon,\delta)}
% H \Sigma_{\hat x^{\mathrm{KF}}\zeta}\Sigma_{\zeta\zeta}^{-1}\Sigma_{\zeta \hat x^{\mathrm{KF}}}H^\top.
% \end{align*}
%  As shown in the proof of Proposition~\ref{prop:kf_eps_delta_generalF}, $\Sigma_{\hat{x}^{\mathrm{KF}}\zeta} = (\Sigma_{\infty} - P^{\mathrm{KF}}) F^{\top}$. The remaining proof follows directly.
 \end{proof}

\section{Proof of Theorem~\ref{thm:noise_lb}}
\label{app:3}

In the following, we split the proof into two steps. (i) In the first step, we provide an upper bound on $\|\Sigma_{\zeta \hat{x}} H^{\top} \|_{F} $. (ii) Then, we apply Lemma~\ref{lem:noise_add} to obtain the lower bound for the noise covariance $\Xi$.

(i) For each $k \ge 0$, the estimator admits the finite-horizon representation
\begin{align*}
\hat x_k
=
A_e^k \hat x_0
+
\sum_{i=0}^{k-1} A_e^i K y_{k-i}.
\end{align*}
Recalling that $\zeta_k = Fx_k$, we obtain
\begin{align*}
\Cov(\zeta_k,\hat x_k)
&=
\Cov(\zeta_k, A_e^k \hat x_0)
+
\sum_{i=0}^{k-1}\Cov(\zeta_k, y_{k-i})\,K^\top (A_e^i)^\top.
\end{align*}
Then, it follows that
\begin{align*}
\Cov(\zeta_k, y_{k-i})
&=
F\,\Cov(x_k,x_{k-i})\,C^\top.
\end{align*}
Therefore, we obtain
\begin{align*}
\Cov(\zeta_k,\hat x_k)
&=
\Cov(\zeta_k, A_e^k \hat x_0)
\\
&+ 
\sum_{i=0}^{k-1}
F\,\Cov(x_k,x_{k-i})\,C^\top K^\top (A_e^i)^\top.
\end{align*}

Under Assumption~\ref{asm:2}, the state process converges to its steady state. In particular, for each fixed $k \ge i \ge 0$,
\begin{align*}
\lim_{k\to\infty}\Cov(x_k,x_{k-i}) = A^i\Sigma_\infty.
\end{align*}
Moreover, since $A_e$ is Schur stable, the term $\Cov(\zeta_k, A_e^k \hat x_0)$ vanishes as $k\to\infty$. Therefore, we have
\begin{align*}
\Sigma_{\zeta\hat x}
=
\Cov(\zeta_\infty,\hat x_\infty)
&=
\sum_{i=0}^{\infty}
F A^i\Sigma_\infty C^\top K^\top (A_e^i)^\top.
\end{align*}
Applying the triangle inequality, we obtain
\begin{align*}
\|\Sigma_{\zeta \hat x} H^\top \|_F
=&
\left\|\sum_{i=0}^{\infty} F A^i\Sigma_\infty C^\top K^\top (A_e^i)^\top H^\top\right\|_F \\
\le &
\sum_{i=0}^{\infty} \|F A^i\Sigma_\infty C^\top K^\top (A_e^i)^\top H^\top \|_F.
\end{align*}
Then, applying the inequality $\|M_i N_i^\top\|_F \le \|M_i\|_F\|N_i\|_2$ to each term with
$M_i = FA^i\Sigma_\infty C^\top,
N_i = H A_e^i K,$
yields
\begin{align*}
\|\Sigma_{\zeta \hat x} H^\top \|_F
\le
\sum_{i=0}^{\infty} \|F A^i\Sigma_\infty C^\top\|_F \|H A_e^i K\|_2.
\end{align*}
By the Cauchy--Schwarz inequality for sequences, we have
\begin{align*}
\|\Sigma_{\zeta \hat x} H^\top \|_F^2
&\le
\Bigl(\sum_{i=0}^{\infty}\|F A^i\Sigma_\infty C^\top\|_F^2\Bigr)
\Bigl(\sum_{i=0}^{\infty}\|H A_e^i K \|_2^2\Bigr).
\end{align*}
Furthermore, it holds that
\begin{align*}
\sum_{i=0}^{\infty}\|F A^i\Sigma_\infty C^\top\|_F^2
&=
\sum_{i=0}^{\infty}
\operatorname{tr}\!\Bigl(
F A^i\Sigma_\infty C^\top C\Sigma_\infty (A^i)^\top F^\top
\Bigr) \\
&=
\operatorname{tr}\!\Bigl(
\sum_{i=0}^{\infty}(A^i)^\top F^\top F A^i
\,\Sigma_\infty C^\top C\Sigma_\infty
\Bigr) \\
&=
\operatorname{tr}\!\bigl(
W_F \Sigma_\infty C^\top C\Sigma_\infty
\bigr).
\end{align*}
In summary, we have
\begin{align}
\label{eq:cross_bound_final}
\|\Sigma_{\zeta \hat x} H^\top \|_F^2
\le
\operatorname{tr}\!\bigl(
W_F \Sigma_\infty C^\top C\Sigma_\infty
\bigr)
\sum_{i=0}^{\infty}\|H A_e^i K \|_2^2.
\end{align}

(ii) Following arguments similar to those in the proof of Theorem~\ref{thm:noise_KF}, it suffices to prove
\begin{align}
\label{eq:thm2_2}
\Xi
&\succeq
\frac{\kappa_\ell (\varepsilon,\delta)}{1-\kappa_\ell (\varepsilon,\delta)}
H \Sigma_{\hat x\zeta}\Sigma_{\zeta\zeta}^{-1}\Sigma_{\zeta \hat x}H^\top.
\end{align}

Combining~\eqref{eq:noise_lb} with~\eqref{eq:cross_bound_final} yields
\begin{align*}
    \Xi \succeq
\frac{\kappa_r(\varepsilon,\delta)}{1-\kappa_r(\varepsilon,\delta)}
\|\Sigma_{\zeta\zeta}^{-1}\|_2
\|\Sigma_{\zeta \hat x} H^\top \|_F^2 I_{s}.
\end{align*}
Moreover, since $r \ge \ell$ by the definition of $\ell$ and $\kappa_r(\varepsilon,\delta) \ge \kappa_\ell(\varepsilon,\delta)$ by the monotonicity of $\mathbf{F}_{\chi^2_\ell}^{-1}$ with respect to the degrees of freedom, it follows that
\begin{align*}
\frac{\kappa_r(\varepsilon,\delta)}{1-\kappa_r(\varepsilon,\delta)}
\ge
\frac{\kappa_\ell(\varepsilon,\delta)}{1-\kappa_\ell(\varepsilon,\delta)}.
\end{align*}
Therefore, it holds that 
\begin{align*}
    \Xi \succeq
\frac{\kappa_{\ell}(\varepsilon,\delta)}{1-\kappa_{\ell}(\varepsilon,\delta)}
\|\Sigma_{\zeta\zeta}^{-1}\|_2
\|\Sigma_{\zeta \hat x} H^\top \|_F^2 I_{s}.
\end{align*}
Applying
\begin{align*}
\lambda_{\max}(MNM^\top)\le \|N\|_2 \|M\|_F^2,
\end{align*}
with $M= \Sigma_{\zeta \hat x} H^\top$ and $N=\Sigma_{\zeta\zeta}^{-1}$, we have
\begin{align*}
\Xi
&\succeq
\frac{\kappa_\ell (\varepsilon,\delta)}{1-\kappa_\ell (\varepsilon,\delta)}
\lambda_{\max}(H \Sigma_{\hat x\zeta}\Sigma_{\zeta\zeta}^{-1}\Sigma_{\zeta \hat x}H^\top)  I_{s},
\end{align*}
which implies~\eqref{eq:thm2_2}.
\QED

\section{Proof of Theorem~\ref{thm:convex_opt}}
\label{app:4}
 The proof is divided into three steps: (i) we show that the estimator~\eqref{eq:estimator} is stable under the choice $K=-X^{-1}U^\top$; (ii) we prove that the estimator with $K=-X^{-1}U^\top$ and $\Xi=\frac{\beta(\varepsilon,\delta)}{s}\operatorname{tr}(W)I_s$ guarantees the required privacy constraint; and (iii) we derive the corresponding upper bound on the estimation error.

(i) Define the estimation error
\begin{align*}
e_k:=x_k-\hat x_k.
\end{align*}
Using \eqref{eq:LTI} and \eqref{eq:estimator}, the error dynamics can be written as
\begin{align*}
e_{k+1}
&=A_e e_k + B_e \eta_k.
\end{align*}
where $A_e = A -KCA$, $B_e:=\begin{bmatrix}(I-KC)B & -KD\end{bmatrix}$, and $\eta_k:=\begin{bmatrix}w_k\\ v_{k+1}\end{bmatrix}$.
Since $U=-K^\top X$, it follows that
\begin{align*}
A^\top X + A^\top C^\top U
&=A^\top X-A^\top C^\top K^\top X
=A_e^\top X. 
\end{align*}
Using this equality, condition \eqref{eq:LMI1} can be rewritten as
\begin{align*}
\begin{bmatrix}
X & A_e^\top X & H^\top\\
X A_e & X & 0\\
H & 0 & I
\end{bmatrix}\succ0.
\end{align*}
Applying the Schur complement first with respect to the $(3,3)$ block and then with respect to the $(2,2)$ block yields
\begin{align*}
X \succ A_e^\top X A_e + H^\top H.
\end{align*}
Therefore, $A_e$ is Schur stable. This completes the proof of part (i).

(ii) We now show that the estimator with $K=-X^{-1}U^\top$ and $\Xi=\frac{\beta(\varepsilon,\delta)}{s}\operatorname{tr}(W)I_s$ satisfies the required privacy constraint.

By the Schur complement, condition~\eqref{eq:LMI3} implies
\begin{align*}
W \succeq U X^{-1}U^\top.
\end{align*}
Since $U=-K^\top X$, we have
\begin{align*}
U X^{-1}U^\top
=K^\top X X^{-1} X K =K^\top X K.
\end{align*}
Hence,
\begin{align*}
W\succeq K^\top X K.
\end{align*}
Taking traces on both sides gives
\begin{align*}
\operatorname{tr}(W)\ge \operatorname{tr}(K^\top X K).
\end{align*}
From part~(i), we already know that
\begin{align*}
X\succeq P,
\end{align*}
where $P$ is the unique positive semidefinite solution of
\begin{align*}
P=A_e^\top P A_e + H^\top H,
\end{align*}
or equivalently,
\begin{align}
\label{eq:series_P}
P=\sum_{i=0}^{\infty}(A_e^i)^\top H^\top H A_e^i.
\end{align}
Therefore, it follows that
\begin{align}
\label{eq:thm4_W}
\operatorname{tr}(W)
\ge \operatorname{tr}(K^\top X K) \ge \operatorname{tr}(K^\top P K).
\end{align}
Using the series representation of $P$ in~\eqref{eq:series_P}, we obtain
\begin{align*}
\operatorname{tr}(K^\top P K)
&=
\operatorname{tr}\!\left(
K^\top
\sum_{i=0}^{\infty}(A_e^i)^\top H^\top H A_e^i
K
\right) \\
&=
\sum_{i=0}^{\infty}
\operatorname{tr}\!\left(
K^\top (A_e^i)^\top H^\top H A_e^i K
\right) \\
&=
\sum_{i=0}^{\infty}\|H A_e^i K\|_F^2.
\end{align*}
This, together with~\eqref{eq:thm4_W}, implies that
\begin{align*}
\operatorname{tr}(W)\ge \sum_{i=0}^{\infty}\|H A_e^i K\|_F^2 \ge  \sum_{i=0}^{\infty}\|H A_e^i K\|_2^2.
\end{align*}
By the sufficient privacy condition established in Theorem~\ref{thm:noise_lb}, it follows that choosing
\begin{align*}
\Xi
=
\frac{\kappa_r (\varepsilon,\delta)\|\Sigma_{\zeta\zeta}^{-1}\|_2}{1-\kappa_r (\varepsilon,\delta)}
\operatorname{tr}\!\bigl(
W_F \Sigma_\infty C^\top C\Sigma_\infty
\bigr)\operatorname{tr}(W)\,I_s
\end{align*}
guarantees $(\varepsilon,\delta)$-PML privacy for the secret variable $\zeta_\infty$.
This is exactly the chosen noise covariance, which completes the proof of part~(ii).

(iii) We now derive the upper bound on the steady-state estimation error. By $U=-K^\top X$, we have
\begin{align*}
XB+U^\top C B
=XB-XKCB
=X(I-KC)B
\end{align*}
and
\begin{align*}
U^\top D =-XKD.
\end{align*}
These two equalities yield
\begin{align*}
\begin{bmatrix}
XB+U^\top C B & U^\top D
\end{bmatrix}
=
X B_e .
\end{align*}
Therefore, condition~\eqref{eq:LMI2} can be rewritten as
\begin{align*}
\begin{bmatrix}
Z & B_e^\top X\\
X B_e & X
\end{bmatrix}\succeq 0.
\end{align*}
Since $X\succ0$, the Schur complement lemma yields
\begin{align*}
Z \succeq B_e^\top X B_e.
\end{align*}
Taking traces and using $X\succeq P$ gives
\begin{align}
\label{eq:thm4_Z}
\operatorname{tr}(Z)
\ge \operatorname{tr}(B_e^\top X B_e) \ge \operatorname{tr}(B_e^\top P B_e).
\end{align}
Let $\Sigma_e:=\lim_{k\to\infty}\Cov(e_k)$.
Since $A_e$ is Schur stable and $\eta_k$ has identity covariance, $\Sigma_e$ is the unique solution of
\begin{align*}
\Sigma_e=A_e\Sigma_e A_e^\top + B_e B_e^\top.
\end{align*}
Using the series representation of $P$ in~\eqref{eq:series_P}, we obtain
\begin{align*}
\operatorname{tr}(B_e^\top P B_e)
&=
\sum_{i=0}^{\infty}\|H A_e^i B_e\|_F^2
=
\operatorname{tr}(H\Sigma_e H^\top).
\end{align*}
This, together with~\eqref{eq:thm4_Z}, implies that
\begin{align*}
\operatorname{tr}(H\Sigma_e H^\top)\le \operatorname{tr}(Z).
\end{align*}

Finally, since
$z_k-\hat z_k =Hx_k-(H\hat x_k+\xi_k) =H(x_k-\hat x_k)-\xi_k =He_k-\xi_k,
$
and $\xi_k$ is independent of $e_k$, we have
\begin{align*}
\lim_{k\to\infty}\Cov(z_k-\hat z_k)
=
H\Sigma_e H^\top + \Xi.
\end{align*}
Taking traces yields
\begin{align*}
\lim_{k\to\infty}\operatorname{tr}\!\bigl(\Cov(z_k-\hat z_k)\bigr)
&=
\operatorname{tr}(H\Sigma_e H^\top)+\operatorname{tr}(\Xi) \\
&\le
\operatorname{tr}(Z)+\operatorname{tr}(\Xi).
\end{align*}
Since
\begin{align*}
\operatorname{tr}(\Xi)=\beta(\varepsilon,\delta)\operatorname{tr}(W),
\end{align*}
it follows that
\begin{align*}
\lim_{k\to\infty}\operatorname{tr}\!\bigl(\Cov(z_k-\hat z_k)\bigr)
\le
\operatorname{tr}(Z)+\beta(\varepsilon,\delta)\operatorname{tr}(W).
\end{align*}
This completes the proof.
\QED

\bibliographystyle{ieeetr}
\bibliography{ref}

\begin{IEEEbiography}{Le Liu}
received the Bachelor of Engineering degree in power and energy engineering and Master of Engineering degree in systems and control
from Northwest A\&F University and Dalian University of Technology, China, in 2019 and 2022, respectively. He is currently working toward the
Ph.D. Degree in Systems and Control in Faculty of Science and Engineering from
University of Groningen, The Netherlands.
His research interests include privacy of control systems, networked control systems and nonlinear systems.
\end{IEEEbiography}

\begin{IEEEbiography}{Yu Kawano}(Senior Member) 
received his M.S. and Ph.D. degrees in Engineering from Osaka University, Japan, in 2011 and 2013, respectively. He was a Postdoctoral Researcher at Kyoto University, Japan, from 2013 to 2016, and at the University of Groningen, the Netherlands, from 2016 to 2019. He was an Associate Professor from 2019 to 2026 and has been a Full Professor since 2026 at the Graduate School of Advanced Science and Engineering, Hiroshima University. He has also held visiting research positions at the Tallinn University of Technology (Estonia), the University of Groningen (the Netherlands), the University of Pavia (Italy), the Indian Institute of Technology Bombay (India), the University of California, Santa Barbara (USA), the Grenoble Alpes University (France). His research interests include nonlinear systems, complex networks, model reduction, and privacy in control systems. He serves as an Associate Editor for IFAC Automatica, Systems \& Control Letters, IEEE Transactions on Systems, Man, and Cybernetics: Systems, the IEEE CSS Conference Editorial Board, and the EUCA Conference Editorial Board.
\end{IEEEbiography}

\begin{IEEEbiography}{Ming Cao} has since 2016 been a professor of networks and robotics with the Engineering and Technology Institute (ENTEG) at the University of Groningen, the Netherlands, where he started as an assistant professor in 2008. Since 2022 he is the director of the Jantina Tammes School of Digital Society, Technology and AI at the same university. He received the Bachelor degree in 1999 and the Master degree in 2002 from Tsinghua University, China, and the Ph.D. degree in 2007 from Yale University, USA, all in electrical engineering. From 2007 to 2008, he was a Research Associate at Princeton University, USA. He worked as a research intern in 2006 at the IBM T. J. Watson Research Center, USA. He is the 2017 and inaugural recipient of the Manfred Thoma medal from the International Federation of Automatic Control (IFAC) and the 2016 recipient of the European Control Award sponsored by the European Control Association (EUCA). He is an IEEE fellow. He is a Senior Editor for Systems and Control Letters, an Associate Editor for IEEE Transactions on Automatic Control, IEEE Transaction of Control of Network Systems and IEEE Robotics \& Automation Magazine, and was an associate editor for IEEE Transactions on Circuits and Systems and IEEE Circuits and Systems Magazine. He is a member of the IFAC Council and a vice chair of the IFAC Technical Committee on Large-Scale Complex Systems. His research interests include autonomous robots and multi-agent systems, complex networks and decision-making processes. 
\end{IEEEbiography}

\end{document}